\pdfoutput=1
\documentclass[11pt]{article}
\usepackage{subfig}
\usepackage{amssymb,amsmath,amsthm}
\usepackage{graphics,epsfig}
\usepackage{algorithmic}
\usepackage{algorithm}
\usepackage[margin=1 in]{geometry}

\usepackage{lipsum}
\usepackage{mathtools}
\usepackage{cuted}
\usepackage{hhline}

\newtheorem{lem}{Lemma}[section]

\newtheorem{thm}[lem]{Theorem}
\newtheorem{obs}[lem]{\textbf{Observation}}

\title{Efficient Approximation Algorithms for Scheduling Coflows with Precedence Constraints in Identical Parallel Networks to Minimize Weighted Completion Time}
\author{Chi-Yeh~Chen 
\\ Department of Computer Science and Information
Engineering, \\ National Cheng Kung University, \\
Taiwan, ROC. \\
chency@csie.ncku.edu.tw.}

\begin{document}

\maketitle
\begin{abstract}
This paper focuses on the problem of coflow scheduling with precedence constraints in identical parallel networks, which is a well-known $\mathcal{NP}$-hard problem. Coflow is a relatively new network abstraction used to characterize communication patterns in data centers. Both flow-level scheduling and coflow-level scheduling problems are examined, with the key distinction being the scheduling granularity. The proposed algorithm effectively determines the scheduling order of coflows by employing the primal-dual method. When considering workload sizes and weights that are dependent on the network topology in the input instances, our proposed algorithm for the flow-level scheduling problem achieves an approximation ratio of $O(\chi)$ where $\chi$ is the coflow number of the longest path in the directed acyclic graph (DAG). Additionally, when taking into account workload sizes that are topology-dependent, the algorithm achieves an approximation ratio of $O(R\chi)$, where $R$ represents the ratio of maximum weight to minimum weight. For the coflow-level scheduling problem, the proposed algorithm achieves an approximation ratio of $O(m\chi)$, where $m$ is the number of network cores, when considering workload sizes and weights that are topology-dependent. Moreover, when considering workload sizes that are topology-dependent, the algorithm achieves an approximation ratio of $O(Rm\chi)$. In the coflows of multi-stage job scheduling problem, the proposed algorithm achieves an approximation ratio of $O(\chi)$. Although our theoretical results are based on a limited set of input instances, experimental findings show that the results for general input instances outperform the theoretical results, thereby demonstrating the effectiveness and practicality of the proposed algorithm.

\end{abstract}

\section{Introduction}\label{sec:introduction}
With the evolution of technology, a large volume of computational demands has become the norm. As personal computing resources are no longer sufficient, cloud computing has emerged as a solution for accessing significant computational resources. With the increasing demand, large-scale data centers have become essential components of cloud computing. In these data centers, the benefits of application-aware network scheduling have been proven, particularly for distributed applications with structured traffic patterns~\cite{Chowdhury2014, Chowdhury2015, Zhang2016, Agarwal2018}. The widespread use of data-parallel computing applications such as MapReduce~\cite{Dean2008}, Hadoop~\cite{Shvachko2010, borthakur2007hadoop}, Dryad~\cite{isard2007dryad}, and Spark~\cite{zaharia2010spark} has led to a proliferation of related applications~\cite{dogar2014decentralized, chowdhury2011managing}.

In these data-parallel applications, tasks can be divided into multiple computational stages and communication stages, which are executed alternately. The computational stages generate a substantial amount of intermediate data (flows) that needs to be transmitted across various machines for further processing during the communication stages. Due to the large number of applications generating significant data transmission requirements, robust data transmission and scheduling capabilities are crucial for data centers. The overall communication pattern within the data center can be abstracted by coflow traffic, representing the interaction of flows between two sets of machines~\cite{Chowdhury2012}.

A coflow refers to a set of interconnected flows, where the completion time of the entire group depends on the completion time of the last flow within the set~\cite{shafiee2018improved}. Previous studies related to coflows~\cite{Chowdhury2014, Chowdhury2015, Zhang2016, huang2016, Agarwal2018, Qiu2015, ahmadi2020scheduling, khuller2016brief, shafiee2018improved, shafiee2021scheduling} have primarily focused on the single-core model~\cite{Huang2020}. However, technological advancements have led to the emergence of data centers that operate on multiple parallel networks in order to improve efficiency~\cite{Singh2015, Huang2020}. One such architecture is the identical or heterogeneous parallel network, where multiple network cores function in parallel, providing combined bandwidth by simultaneously serving traffic.

\begin{table*}[!ht]
\caption{Theoretical Results}
\centering
\begin{tabular}{|c|c|c|c|c|}
\hline
       Model                      & TDWS  & TDW        & Approximation Ratio &  \\ \hhline{|=|=|=|=|=|}
 flow-level scheduling problem    &   $\surd$          &    $\surd$                        &  $O(\chi)$       &  Thms~\ref{thm:thm1} \& \ref{thm:thm2}\\ \hline
 flow-level scheduling problem    &   $\surd$          &                                   &  $O(R\chi)$    &  Thms~\ref{thm:thm1-1} \& \ref{thm:thm2-1}\\ \hline
 coflow-level scheduling problem  &   $\surd$          &    $\surd$                        &  $O(m\chi)$    &  Thms~\ref{thm:thm21} \& \ref{thm:thm22}\\ \hline
 coflow-level scheduling problem  &   $\surd$          &                                   &  $O(Rm\chi)$   &  Thms~\ref{thm:thm21-1} \& \ref{thm:thm22-1}\\ \hline
 multi-stage job scheduling problem  &                 &                                   &  $O(\chi)$   &  Thm~\ref{thm4:thm11}\\ \hline
\end{tabular}
\label{tab:results}
\end{table*}

This study addresses the problem of coflow scheduling with precedence constraints in identical parallel networks. The objective is to schedule these coflows in the parallel networks in a way that minimizes the weighted total completion time of coflows. We consider both flow-level scheduling and coflow-level scheduling. In the flow-level scheduling problem, flows within a coflow can be distributed across different network cores. Conversely, in the coflow-level scheduling problem, all flows within a coflow are required to be transmitted in the same network core. The key difference between these two problems lies in their scheduling granularity. The coflow-level scheduling problem, being a coarse-grained scheduling, can be quickly solved but yields relatively poorer results. On the other hand, the flow-level scheduling problem, being a fine-grained scheduling, takes more time to solve but produces superior scheduling results. It is worth noting that, although these two problems exhibit differences in time complexity when solved using linear programming, in the case of the flow-level scheduling problem using the primal-dual method, the decision of scheduling flows is transformed into the decision of scheduling coflows. This transformation leads to the solving time being equivalent to that of the coflow-level scheduling problem.

\subsection{Related Work}
The concept of coflow abstraction was initially introduced by Chowdhury and Stoica~\cite{Chowdhury2012} to characterize communication patterns within data centers. The scheduling problem for coflows has been proven to be strongly $\mathcal{NP}$-hard, indicating the need for efficient approximation algorithms rather than exact solutions. Due to the easy reduction of the concurrent open shop problem to coflow scheduling, where only the diagonal elements of the demand matrix have values, solving the concurrent open shop problem within a factor better than $2-\epsilon$ is $\mathcal{NP}$-hard~\cite{Bansal2010, Sachdeva2013}, implying the hardness of the coflow scheduling problem as well. 

Since the proposal of the coflow abstraction, extensive research has been conducted on coflow scheduling~\cite{Chowdhury2014, Chowdhury2015, Qiu2015, zhao2015rapier, shafiee2018improved, ahmadi2020scheduling}. Qiu~\textit{et al.}\cite{Qiu2015} presented the first deterministic polynomial-time approximation algorithm with an ratio of $\frac{67}{3}$. Subsequently, Ahmadi \textit{et al.}~\cite{ahmadi2020scheduling} proved that the technique proposed by Qiu \textit{et al.}\cite{Qiu2015} actually yields only a deterministic $\frac{76}{3}$-approximation algorithm for coflow scheduling with release times.
Khuller \textit{et al.}~\cite{khuller2016brief} also proposed an approximation algorithm for coflow scheduling with arbitrary release times, achieving a ratio of $12$.
Recent research by Shafiee and Ghaderi~\cite{shafiee2018improved} has resulted in an impressive approximation algorithm for the coflow scheduling problem, achieving an approximation ratio of $5$. Additionally, Ahmadi \textit{et al.}~\cite{ahmadi2020scheduling} have made significant contributions to this field by proposing a primal-dual algorithm that enhances the computational efficiency of coflow scheduling.

In the coflow scheduling problem within a heterogeneous parallel network, Huang \textit{et al.}~\cite{Huang2020} introduced an $O(m)$-approximation algorithm, where $m$ represents the number of network cores. On the other hand, Tian \textit{et al.}~\cite{Tian18} were the first to propose the problem of scheduling coflows of multi-stage jobs, and they provided a $O(N)$-approximation algorithm, where $N$ represents the number of servers in the network. Furthermore, Shafiee and Ghaderi~\cite{shafiee2021scheduling} proposed a polynomial-time algorithm that achieves an approximation ratio of $O(\tilde{\chi} \log(N)/\log(\log(N)))$, where $\tilde{\chi}$ denotes the maximum number of coflows in a job.

\subsection{Our Contributions}
This paper focuses on addressing the problem of coflow scheduling with precedence constraints in identical parallel networks and presents a range of algorithms and corresponding results. The specific contributions of this study are outlined below:

\begin{itemize}
\item When considering workload sizes and weights that are dependent on the network topology in the input instances, the proposed algorithm for the flow-level scheduling problem achieves an approximation ratio of $O(\chi)$ where $\chi$ is the coflow number of the longest path in the directed acyclic graph (DAG). 

\item When taking into account workload sizes that are topology-dependent, the proposed algorithm for flow-level scheduling problem achieves an approximation ratio of $O(R\chi)$, where $R$ represents the ratio of maximum weight to minimum weight.

\item For the coflow-level scheduling problem, the proposed algorithm achieves an approximation ratio of $O(m\chi)$, where $m$ is the number of network cores, when considering workload sizes and weights that are topology-dependent.

\item When considering workload sizes that are topology-dependent, the algorithm for the coflow-level scheduling problem achieves an approximation ratio of $O(Rm\chi)$.

\item In the coflows of multi-stage job scheduling problem, the proposed algorithm achieves an approximation ratio of $O(\chi)$.
\end{itemize}

A summary of our theoretical findings is provided in Table~\ref{tab:results} where TDWS stands for topology-dependent workload sizes, while TDW stands for topology-dependent weights.

\subsection{Organization}
The structure of this paper is outlined as follows. In Section~\ref{sec:Preliminaries}, an introduction is provided, covering fundamental notations and preliminary concepts that will be referenced in subsequent sections. Following that, the primary algorithms are presented in the following sections: Section~\ref{sec:Algorithm1} provides an overview of the algorithm addressing the flow-level scheduling problem, while Section~\ref{sec:Algorithm2} elaborates on the algorithm designed for the coflow-level scheduling problem. To address the scheduling problem for the coflows of multi-stage jobs, our algorithm is discussed in Section~\ref{sec:Algorithm6}. In Section~\ref{sec:Results}, a comparative analysis is conducted to evaluate the performance of our proposed algorithms in comparison to the previous algorithm. Lastly, in Section~\ref{sec:Conclusion}, our findings are summarized and meaningful conclusions are drawn.

\section{Notation and Preliminaries}\label{sec:Preliminaries}
The identical parallel network consists of a collection of $m$ non-blocking switches, each with dimensions of $N \times N$. These switches form the infrastructure of the network, where $N$ input links are connected to $N$ source servers, and $N$ output links are connected to $N$ destination servers. These switches serve as practical and intuitive models for the network core. Network architectures such as Fat-tree or Clos~\cite{al2008scalable, greenberg2009vl2} can be employed to construct networks that provide complete bisection bandwidth. In this configuration, each switch's $i$-th input port is connected to the $i$-th source server, and the $j$-th output port is connected to the $j$-th destination server. Consequently, each source server (or destination server) has $m$ simultaneous uplinks (or downlinks), where each link may consist of multiple physical connections in the actual network topology~\cite{Huang2020}. Let $\mathcal{I}$ denote the set of source servers, and $\mathcal{J}$ denote the set of destination servers. The network core can be visualized as a bipartite graph, with $\mathcal{I}$ on one side and $\mathcal{J}$ on the other. For simplicity, we assume that all network cores are identical, and the links within each core have the same capacity or speed.

A coflow is a collection of independent flows, and its completion time of a coflow is determined by the completion time of the last flow in the set, making it a critical metric for evaluating the efficiency of data transfers. The demand matrix $D^{(k)}=\left(d_{i,j,k}\right)_{i,j=1}^{N}$ represents the specific data transfer requirements within coflow $k$. Each entry $d_{i,j,k}$ in the matrix corresponds to the size of the flow that needs to be transmitted from input $i$ to output $j$ within the coflow. In the context of identical network cores, the flow size can be interpreted as the transmission time, as all cores possess the same capacity or speed. This simplification allows for easier analysis and optimization of coflow scheduling algorithms. To facilitate efficient management and routing of flows, each flow is identified by a triple $(i, j, k)$, where $i$ represents the source node, $j$ represents the destination node, and $k$ corresponds to the coflow. This identification scheme enables precise tracking and control of individual flows within the parallel network.
Furthermore, we assume that flows are composed of discrete data units, resulting in integer sizes. For simplicity, we assume that all flows within a coflow are simultaneously initiated, as demonstrated in~\cite{Qiu2015}.

This paper investigates the problem of coflow scheduling with release times and precedence constraints. The problem involves a set of coflows denoted by $\mathcal{K}$, where coflow $k$ is released into the system at time $r_k$. The completion time of coflow $k$, denoted as $C_k$, represents the time required for all its flows to finish processing. Each coflow $k\in \mathcal{K}$ is assigned a positive weight $w_k$. Let $R$ be the ratio between the maximum weight and the minimum weight. The relationships between coflows can be modeled using a directed acyclic graph (DAG) $G=(\mathcal{K}, E)$, where an arc $(k', k)\in E$ and $k', k\in \mathcal{K}$ indicate that all flows of coflow $k'$ must be completed before any flow of coflow $k$ can be scheduled. This relationship is denoted as $k'\prec k$. The DAG has a coflow number of $\chi$, which represents the length of the longest path in the DAG. The objective is to schedule coflows in an identical parallel network, considering the precedence constraints, in order to minimize the total weighted completion time of the coflows, denoted as $\sum_{k\in \mathcal{K}} w_{k}C_{k}$. For clarity, different subscript symbols are used to represent different meanings of the same variables. Subscript $i$ represents the index of the source (or input port), subscript $j$ represents the index of the destination (or output port), and subscript $k$ represents the index of the coflow. For instance, $\mathcal{F}_i$ denotes the set of flows with source $i$, and $\mathcal{F}_j$ represents the set of flows with destination $j$. The symbols and terminology used in this paper are summarized in Table~\ref{tab:notations}.

\begin{table}[!ht]
\caption{Notation and Terminology}
    \centering
        \begin{tabular}{||c|p{4.5in}||}
    \hline
     $m$      & The number of network cores.          \\
    \hline    
     $N$      & The number of input/output ports.         \\
    \hline
     $n$      & The number of coflows.         \\
    \hline
     $\mathcal{I}, \mathcal{J}$ & The source server set and the destination server set.         \\
    \hline    
     $\mathcal{K}$ & The set of coflows.         \\
    \hline
     $D^{(k)}$     & The demand matrix of coflow $k$. \\
    \hline    
     $d_{i,j,k}$     & The size of the flow to be transferred from input $i$ to output $j$ in coflow $k$.   \\
    \hline     
     $C_k$     & The completion time of coflow $k$.   \\
    \hline     
     $C_{i,j,k}$ & The completion time of flow $(i, j, k)$. \\
    \hline     
     $r_k$     & The released time of coflow $k$.  \\
    \hline     
     $w_{k}$   &  The weight of coflow $k$. \\
    \hline     
		 $\mathcal{F}_{i}$ & $\mathcal{F}_{i}=\left\{(i, j, k)| d_{i,j,k}>0, \forall k\in \mathcal{K}, \forall j\in \mathcal{J} \right\}$ is the set of flows with source $i$. \\
		\hline 					
		 $\mathcal{F}_{j}$ & $\mathcal{F}_{j}=\left\{(i, j, k)| d_{i,j,k}>0, \forall k\in \mathcal{K}, \forall i\in \mathcal{I} \right\}$ is the set of flows with destination $j$. \\
		\hline 					
		 $d(S), d^2(S)$ & $d(S)=\sum_{(i, j, k)\in S} d_{i,j,k}$ and $d^2(S)=\sum_{(i, j, k)\in S} d_{i,j,k}^{2}$ for any subset $S\subseteq \mathcal{F}_{i}$ (or $S\subseteq \mathcal{F}_{j}$). \\
		\hline 					
		 $f(S)$ & $f(S) = \frac{d(S)^2+ d^2(S)}{2m}$ for any subset $S\subseteq \mathcal{F}_{i}$ (or $S\subseteq \mathcal{F}_{j}$). \\
		\hline 				
			$L_{i,S,k}$ & $L_{i,S,k}=\sum_{(i',j',k')\in S|i'=i,k'=k}d_{i',j',k'}$ is the total load on input port $i$ for coflow $k$ in the set $S$. \\
		\hline 				
			$L_{j,S,k}$ & $L_{j,S,k}=\sum_{(i',j',k')\in S|j'=j,k'=k}d_{i',j',k'}$ is the total load on output port $j$ for coflow $k$ in the set $S$. \\
		\hline 				
			$L_{i,k}$ & $L_{i,k}=\sum_{j\in \mathcal{J}} d_{i,j,k}$ is the total load of flows from the coflow $k$ at input port $i$. \\
		\hline 				
			$L_{j,k}$ & $L_{j,k}=\sum_{i\in \mathcal{I}} d_{i,j,k}$ is the total load of flows from the coflow $k$ at output port $j$. \\
		\hline 				
			$L_{i}, L_{j}$ & $L_{i} = \sum_{k\in \mathcal{K}}L_{i,k}$ and $L_{j} = \sum_{k\in \mathcal{K}}L_{j,k}$. \\
		\hline 				
			$S_{i,k}$ & $S_{i,k}$ is the set of flows from the first $k$ coflows at input port $i$. \\
		\hline 				
			$S_{j,k}$ & $S_{j,k}$ is the set of flows from the first $k$ coflows at output port $j$. \\
		\hline 				
			$f_{i}(S)$ & $f_{i}(S) = \frac{\sum_{k\in S} L_{i,k}^2+\left(\sum_{k\in S} L_{i,k}\right)^2}{2m}$ for any subset $S\subseteq \mathcal{K}$. \\
		\hline 				
			$f_{j}(S)$ & $f_{j}(S) = \frac{\sum_{k\in S} L_{j,k}^2+\left(\sum_{k\in S} L_{j,k}\right)^2}{2m}$ for any subset $S\subseteq \mathcal{K}$. \\
		\hline 				
			$S_{k}$ & $S_{k}=\left\{1, 2, \ldots, k\right\}$ is the set of first $k$ coflows. \\
		\hline 				
			$L_{i}(S_{k}), L_{j}(S_{k})$ & $L_{i}(S_{k})=\sum_{k'\leq k} L_{i, k'}$ and $L_{j}(S_{k})=\sum_{k'\leq k} L_{j, k'}$. \\
		\hline 				
			$\mu_1(k)$ & $\mu_1(k)$ is the input port in $S_{k}$ with the highest load. \\
		\hline 				
			$\mu_2(k)$ & $\mu_2(k)$ is the output port in $S_{k}$ with the highest load. \\
		\hline 				
			$\chi$      & The coflow number of the longest path in the DAG. \\
    \hline     	
			$R$      & The ratio between the maximum weight and the minimum weight. \\
    \hline     	
        \end{tabular}
    \label{tab:notations}
\end{table}


\section{Approximation Algorithm for the Flow-level Scheduling Problem}\label{sec:Algorithm1}
This section focuses on the flow-level scheduling problem, which allows for the transmission of different flows within a coflow through distinct network cores. We assume that coflows are transmitted at the flow level, ensuring that the data within a flow is allocated to the same core. We define $\mathcal{F}_{i}$ as the collection of flows with source $i$, represented by $\mathcal{F}_{i}=\left\{(i, j, k)| d_{i,j,k}>0, \forall k\in \mathcal{K}, \forall j\in \mathcal{J} \right\}$, and $\mathcal{F}_{j}$ as the set of flows with destination $j$, given by $\mathcal{F}_{j}=\left\{(i, j, k)| d_{i,j,k}>0, \forall k\in \mathcal{K}, \forall i\in \mathcal{I} \right\}$. For any subset $S\subseteq \mathcal{F}_{i}$ (or $S\subseteq \mathcal{F}_{j}$), we define $d(S)=\sum_{(i, j, k)\in S} d_{i,j,k}$ as the sum of data size over all flows in $S$ and $d^2(S)=\sum_{(i, j, k)\in S} d_{i,j,k}^{2}$ as the sum of squares of data size over all flows in $S$. Additionally, we introduce the function $f(S)$ as follows:
\begin{eqnarray*}
f(S) = \frac{d(S)^2+ d^2(S)}{2m}.
\end{eqnarray*}
The flow-level scheduling problem can be formulated as a linear programming relaxation, which is expressed as follows:
\begin{subequations}\label{coflow:main}
\begin{align}
& \text{min}  && \sum_{k \in \mathcal{K}} w_{k} C_{k}     &   & \tag{\ref{coflow:main}} \\
& \text{s.t.} && C_{k} \geq C_{i,j,k}, && \forall k\in \mathcal{K}, \forall i\in \mathcal{I}, \forall j\in \mathcal{J} \label{coflow:a} \\
&  && C_{i,j,k}\geq r_k+d_{i,j,k}, && \forall k\in \mathcal{K}, \forall i\in \mathcal{I}, \forall j\in \mathcal{J} \label{coflow:b} \\
&  && C_{i,j,k}\geq C_{k'}+d_{i,j,k}, && \forall k, k'\in \mathcal{K}:k'\prec k, \notag \\
&  &&                             &&  \forall i\in \mathcal{I}, \forall j\in \mathcal{J} \label{coflow:prec:e} \\
&  && \sum_{(i, j, k)\in S}d_{i,j,k}C_{i,j,k}\geq f(S),&& \forall i\in \mathcal{I}, \forall S\subseteq \mathcal{F}_{i} \label{coflow:c} \\
&  && \sum_{(i, j, k)\in S}d_{i,j,k}C_{i,j,k}\geq f(S),&& \forall j\in \mathcal{J}, \forall S\subseteq \mathcal{F}_{j} \label{coflow:d} 
\end{align}
\end{subequations}

In the linear program (\ref{coflow:main}), the variable $C_{k}$ represents the completion time of coflow $k$ in the schedule, and $C_{i,j,k}$ denotes the completion time of flow $(i, j, k)$. Constraint (\ref{coflow:a}) specifies that the completion time of coflow $k$ is bounded by the completion times of all its flows, ensuring that no flow finishes after the coflow. Constraint (\ref{coflow:b}) guarantees that the completion time of any flow $(i, j, k)$ is at least its release time $r_k$ plus the time required for its transmission. To capture the precedence constraints among coflows, constraint (\ref{coflow:prec:e}) indicates that all flows of coflow $k'$ must be completed before any flow of coflow $k$ can be scheduled. Constraints (\ref{coflow:c}) and (\ref{coflow:d}) introduce lower bounds on the completion time variables at the input and output ports, respectively.

We define $L_{i,S,k}$ as the sum of the loads on input port $i$ for coflow $k$ in the set $S$. Similarly, $L_{j,S,k}$ represents the sum of the loads on output port $j$ for coflow $k$ in the set $S$. To formulate the dual linear program, we have the following expressions:
\begin{align*}
L_{i,S,k} &=\sum_{(i',j',k')\in S|i'=i,k'=k}d_{i',j',k'}, \\
L_{j,S,k} &=\sum_{(i',j',k')\in S|j'=j,k'=k}d_{i',j',k'}.
\end{align*}
The dual linear program is given by
\begin{subequations}\label{coflow:dual}
\begin{align}
& \text{max}  && \sum_{k \in \mathcal{K}}\sum_{i \in \mathcal{I}}\sum_{j \in \mathcal{J}} \alpha_{i, j, k}(r_k+d_{i,j,k}) &   &\notag\\
&   && +\sum_{i \in \mathcal{I}}\sum_{S \subseteq \mathcal{F}_{i}}\beta_{i,S} f(S)     &   & \notag\\
&   && +\sum_{j \in \mathcal{J}}\sum_{S \subseteq \mathcal{F}_{j}}\beta_{j,S} f(S)     &   & \notag \\
&   && + \sum_{(k', k) \in E}\sum_{i \in \mathcal{I},j \in \mathcal{J}}\gamma_{k', i, j, k} d_{i,j,k}     &   & \tag{\ref{coflow:dual}} \\
& \text{s.t.} && \sum_{i \in \mathcal{I}}\sum_{j \in \mathcal{J}} \alpha_{i, j, k} &   & \notag\\
&   && +\sum_{i \in \mathcal{I}}\sum_{S\subseteq \mathcal{F}_{i}}\beta_{i,S}L_{i,S,k} &   &\notag\\
&   && +\sum_{j \in \mathcal{J}}\sum_{S\subseteq \mathcal{F}_{j}}\beta_{j,S}L_{j,S,k} &   &\notag\\
&   && +\sum_{(k',k)\in E}\sum_{i \in \mathcal{I},j \in \mathcal{J}}\gamma_{k', i, j, k} &   &\notag\\
&   && -\sum_{(k,k')\in E}\sum_{i \in \mathcal{I},j \in \mathcal{J}}\gamma_{k, i, j, k'}\leq w_{k}, && \forall k\in \mathcal{K} \label{coflow:dual:a} \\
&  && \alpha_{i, j, k} \geq 0, && \forall k\in \mathcal{K}, \forall i\in \mathcal{I}, \notag\\
&  &&                          && \forall j\in \mathcal{J} \label{coflow:dual:b} \\
&  && \beta_{i, S}\geq 0,   &&  \forall i\in \mathcal{I}, \forall S\subseteq \mathcal{F}_{i} \label{coflow:dual:c} \\
&  && \beta_{j, S}\geq 0,   &&  \forall j\in \mathcal{J}, \forall S\subseteq \mathcal{F}_{j} \label{coflow:dual:d} \\
&  && \gamma_{k', i, j, k}\geq 0,   &&  \forall (k', k)\in E, \forall i\in \mathcal{I}, \notag\\
&  &&                          &&\forall j\in \mathcal{J}\label{coflow:dual:e}
\end{align}
\end{subequations}

It is important to note that each flow $(i, j, k)$ is associated with a dual variable $\alpha_{i, j, k}$, and for every coflow $k$, there exists a corresponding constraint. Additionally, for any subset $S \subseteq \mathcal{F}_{i}$ (or $S \subseteq \mathcal{F}_{j}$) of flows, there exists a dual variable $\beta_{i, S}$ (or $\beta_{j, S}$). To facilitate the analysis and design of algorithms, we define $\gamma_{k', k}$ as the sum of $\gamma_{k', i, j, k}$ over all input ports $i$ and output ports $j$ in their respective sets $\mathcal{I}$ and $\mathcal{J}$:
\begin{equation*}
\gamma_{k', k}=\sum_{i \in \mathcal{I},j \in \mathcal{J}}\gamma_{k', i, j, k}.
\end{equation*}
Significantly, it should be emphasized that the cost of any feasible dual solution provides a lower bound for $OPT$, which represents the cost of an optimal solution.
This implies that the cost attained by any valid dual solution ensures that $OPT$ cannot be less than that. In other words, if we obtain a feasible dual solution with a certain cost, we can be certain that the optimal solution, which represents the best possible cost, will not have a lower cost than the one achieved by the dual solution. 

The primal-dual algorithm, as depicted in Appendix~\ref{appendix:a}, Algorithm ~\ref{Alg_dual}, is inspired by the research of Davis \textit{et al.} \cite{DAVIS2013121} and Ahmadi \textit{et al.} \cite{ahmadi2020scheduling}, respectively. This algorithm constructs a feasible schedule iteratively, progressing from right to left, determining the processing order of coflows. Starting from the last coflow and moving towards the first, each iteration makes crucial decisions in terms of increasing dual variables $\alpha$, $\beta$ or $\gamma$. The guidance for these decisions is provided by the dual linear programming (LP) formulation. The algorithm offers a space complexity of $O(Nn)$ and a time complexity of $O(n^2)$, where $N$ represents the number of input/output ports, and $n$ represents the number of coflows.

Consider a specific iteration in the algorithm. At the beginning of this iteration, let $\mathcal{K}$ represent the set of coflows that have not been scheduled yet, and let $k$ denote the coflow with the largest release time. In each iteration, a decision must be made regarding whether to increase dual variables $\alpha$, $\beta$ or $\gamma$. 

If the release time $r_k$ is significantly large, increasing the $\alpha$ dual variable results in substantial gains in the objective function value of the dual problem. On the other hand, if $L_{\mu_1(r)}$ (or $L_{\mu_2(r)}$ if $L_{\mu_2(r)}\geq L_{\mu_1(r)}$) is large, raising the $\beta$ variable leads to substantial improvements in the objective value. Let $\kappa$ be a constant that will be optimized later.

If $r_k>\frac{\kappa\cdot L_{\mu_1(r)}}{m}$ (or $r_k>\frac{\kappa\cdot L_{\mu_2(r)}}{m}$ if $L_{\mu_2(r)}\geq L_{\mu_1(r)}$), the $\alpha$ dual variable is increased until the dual constraint for coflow $k$ becomes tight. Consequently, coflow $k$ is scheduled to be processed as early as possible and before any previously scheduled coflows.
In the case where $r_k\leq \frac{\kappa\cdot L_{\mu_1(r)}}{m}$ (or $r_k\leq\frac{\kappa\cdot L_{\mu_2(r)}}{m}$ if $L_{\mu_2(r)}\geq L_{\mu_1(r)}$), the dual variable $\beta_{\mu_1(r),\mathcal{G}_i}$ (or $\beta_{\mu_2(r),\mathcal{G}_j}$ if $L{\mu_2(r)}\geq L_{\mu_1(r)}$) is increased until the dual constraint for coflow $k'$ becomes tight. 

In this step, we begin by identifying a candidate coflow, denoted as $k'$, with the minimum value of $\beta$. We then examine whether this coflow still has unscheduled successors. If it does, we continue traversing down the chain of successors until we reach a coflow that has no unscheduled successors, which we will refer to as $t_1$.
Once we have identified coflow $t_1$, we set its $\beta$ and $\gamma$ values such that the dual constraint for coflow $t_1$ becomes tight. Moreover, we ensure that the $\beta$ value of coflow $t_1$ matches that of the candidate coflow $k'$.

\begin{algorithm}
\caption{Flow-Driven-List-Scheduling}
    \begin{algorithmic}[1]
				\STATE Both $load_{I}$ and $load_{O}$ are initialized to zero and $\mathcal{A}_h=\emptyset$ for all $h\in [1, m]$
				\FOR{$k=1, 2, \ldots, n$} \label{alg1-2}
				\FOR{every flow $(i, j, k)$ in non-increasing order of $d_{i, j, k}$, breaking ties arbitrarily}
				    \STATE $h^*=\arg \min_{h\in [1, m]}\left(load_{I}(i,h)+load_{O}(j,h)\right)$
						\STATE $\mathcal{A}_{h^*}=\mathcal{A}_{h^*}\cup \left\{(i, j, k)\right\}$
						\STATE $load_{I}(i,h^*)=load_{I}(i,h^*)+d_{i, j, k}$ and $load_{O}(j,h^*)=load_{O}(j,h^*)+d_{i, j, k}$
				\ENDFOR
				\ENDFOR \label{alg1-3}
				\FOR{each $h\in [1, m]$ do in parallel} \label{alg1-4}
				    \STATE wait until the first coflow is released
						\WHILE{there is some incomplete flow}
						    \FOR{$k'=1, 2, \ldots, n$}
                \FOR{every ready, released and incomplete flow $(i, j, k=k')\in \mathcal{A}_{h}$ in non-increasing order of $d_{i, j, k}$, breaking ties arbitrarily}
										\IF{the link $(i, j)$ is idle}
										    \STATE schedule flow $f$\label{alg1-1}
										\ENDIF
								\ENDFOR
								\ENDFOR
								\WHILE{no new flow is ready, completed or released}
								    \STATE transmit the flows that get scheduled in line \ref{alg1-1} at maximum rate 1.
								\ENDWHILE
						\ENDWHILE
				\ENDFOR\label{alg1-5}
   \end{algorithmic}
\label{Alg1}
\end{algorithm}

The flow-driven-list-scheduling algorithm, as depicted in Algorithm~\ref{Alg1}, leverages a list scheduling rule to determine the order of coflows to be scheduled. In order to provide a clear and consistent framework, we assume that the coflows have been pre-ordered based on the permutation generated by Algorithm~\ref{Alg_dual}, where $\sigma(k)=k$ for all $k\in \mathcal{K}$. Thus, the coflows are scheduled sequentially in this predetermined order.
Within each coflow, the flows are scheduled based on a non-increasing order of their sizes, breaking ties arbitrarily. Specifically, for every flow $(i, j, k)$, the algorithm identifies the least loaded network core, denoted as $h^*$, and assigns the flow $(i, j, k)$ to this core.
The algorithm's steps involved in this assignment process are outlined in lines \ref{alg1-2}-\ref{alg1-3}. 

A flow is deemed "ready" for scheduling only when all of its predecessors have been fully transmitted. The algorithm then proceeds to schedule all the flows that are both ready and have been released but remain incomplete. These scheduling steps, encapsulated in lines \ref{alg1-4}-\ref{alg1-5}, have been adapted from the work of Shafiee and Ghaderi~\cite{shafiee2018improved}.

\subsection{Analysis}
In this section, we present a comprehensive analysis of the proposed algorithm, establishing its approximation ratios. Specifically, we demonstrate that the algorithm achieves an approximation ratio of $O(\chi)$ when considering workload sizes and weights that are topology-dependent in the input instances. Additionally, when considering workload sizes that are topology-dependent in the input instances, the algorithm achieves an approximation ratio of $O(R\chi)$ where $R$ is the ratio of maximum weight to minimum weight. It is crucial to note that our analysis assumes that the coflows are arranged in the order determined by the permutation generated by Algorithm~\ref{Alg_dual}, where $\sigma(k)=k$ for all $k\in \mathcal{K}$.

Let $S_k=\left\{1, 2, \ldots, k\right\}$ denote the set of the first $k$ coflows. Furthermore, we define $S_{i,k}$ as the set of flows from the first $k$ coflows at input port $i$. Formally, $S_{i,k}$ is defined as follows:
\begin{eqnarray*}
S_{i,k}=\left\{(i, j, k')| d_{i,j,k'}>0, \forall k'\in \left\{1,\ldots,k\right\}, \forall j\in \mathcal{J} \right\}.
\end{eqnarray*}
Similarly, $S_{j,k}$ represents the set of flows from the first $k$ coflows at output port $j$, defined as:
\begin{eqnarray*}
S_{j,k}=\left\{(i, j, k')| d_{i,j,k'}>0, \forall k'\in \left\{1,\ldots,k\right\}, \forall i\in \mathcal{I} \right\}.
\end{eqnarray*}
Let $\beta_{i,k}=\beta_{i,S_{i,k}}$ and $\beta_{j,k}=\beta_{j,S_{j,k}}$. These variables capture the dual variables associated with the sets $S_{i,k}$ and $S_{j,k}$.

Moreover, we introduce the notation $\mu_1(k)$ to denote the input port with the highest load in $S_{k}$, and $\mu_2(k)$ to represent the output port with the highest load in $S_{k}$. Recall that $d(S)$ represents the sum of loads for all flows in a subset $S$. Therefore, $d(S_{i,k})$ corresponds to the total load of flows from the first $k$ coflows at input port $i$, and $d(S_{j,k})$ corresponds to the total load of flows from the first $k$ coflows at output port $j$.

Finally, let $L_{i,k}=\sum_{j\in \mathcal{J}} d_{i,j,k}$ denote the total load of flows from coflow $k$ at input port $i$, and $L_{j,k}=\sum_{i\in \mathcal{I}} d_{i,j,k}$ denote the total load of flows from coflow $k$ at output port $j$.
Let us begin by presenting several key observations regarding the primal-dual algorithm.
\begin{obs}\label{obs:1}
The following statements hold.

\begin{enumerate}
\item Every nonzero $\beta_{i,S}$ can be written as $\beta_{\mu_1(k),k}$ for some coflow $k$. \label{obs:1-1}
\item Every nonzero $\beta_{j,S}$ can be written as $\beta_{\mu_2(k),k}$ for some coflow $k$. \label{obs:1-2}
\item For every set $S_{\mu_1(k),k}$ that has a nonzero $\beta_{\mu_1(k),k}$ variable, if $k' \leq k$ then $r_{k'}\leq \frac{\kappa\cdot d(S_{\mu_1(k),k})}{m}$. \label{obs:1-3}
\item For every set $S_{\mu_2(k),k}$ that has a nonzero $\beta_{\mu_2(k),k}$ variable, if $k' \leq k$ then $r_{k'}\leq \frac{\kappa\cdot d(S_{\mu_2(k),k})}{m}$. \label{obs:1-4}
\item For every coflow $k$ that has a nonzero $\alpha_{\mu_1(k), 1, k}$, $r_{k}>\frac{\kappa\cdot d(S_{\mu_1(k),k})}{m}$. \label{obs:1-5}
\item For every coflow $k$ that has a nonzero $\alpha_{1, \mu_2(k), k}$, $r_{k}>\frac{\kappa\cdot d(S_{\mu_2(k),k})}{m}$. \label{obs:1-6}
\item For every coflow $k$ that has a nonzero $\alpha_{\mu_1(k), 1, k}$ or a nonzero $\alpha_{1, \mu_2(k), k}$, if $k'\leq k$ then $r_{k'}\leq r_{k}$. \label{obs:1-7}
\end{enumerate}
\end{obs}
The validity of each of the aforementioned observations can be readily verified and directly inferred from the steps outlined in Algorithm~\ref{Alg_dual}. 
\begin{obs}\label{obs:2}
For any subset $S$, we have that $d(S)^2\leq 2m\cdot f(S)$.
\end{obs} 

\begin{lem}\label{lem:lem1}
Let $C_{k}$ represent the completion time of coflow $k$ when scheduled according to Algorithm~\ref{Alg1}. For any coflow $k$, we have $C_{k}\leq a\cdot \max_{k'\leq k}r_{k'}+\chi\left(\frac{d(S_{\mu_1(k),k})+d(S_{\mu_2(k),k})}{m}\right)+(1-\frac{2}{m})C_{k}^*$, where $a=0$ signifies the absence of release times, and $a=1$ indicates the presence of arbitrary release times.
\end{lem}
\begin{proof}
First, let's consider the case where there is no release time and no precedence constraints. In this case, the completion time bound for each coflow can be expressed by the following inequality:
\begin{eqnarray*}
\hat{C}_{k} & \leq & \frac{1}{m}d(S_{\mu_1(k), k}) + \frac{1}{m}d(S_{\mu_2(k),k})+\left(1-\frac{2}{m}\right) \max_{i, j} d_{i,j,k} \label{lem1:eq1}
\end{eqnarray*}

Now, let $v_1v_2\cdots v_f$ be the longest path of coflow $k$, where $v_f=k$. Then, we can derive the following inequalities:
\begin{eqnarray*}
C_{k}   & \leq & \sum_{q=1}^{f} \hat{C}_{v_q} \label{lem4:eq1}\\
        & \leq & \sum_{q=1}^{f} \frac{1}{m}d(S_{\mu_1(q), q}) + \frac{1}{m}d(S_{\mu_2(q),q}) +\left(1-\frac{2}{m}\right) \max_{i, j} d_{i,j,q} \label{lem1:eq2}\\
		    & \leq & \sum_{q=1}^{f} \frac{1}{m}d(S_{\mu_1(k), k}) + \frac{1}{m}d(S_{\mu_2(k),k}) +\left(1-\frac{2}{m}\right) \max_{i, j} d_{i,j,q} \label{lem1:eq3}\\
		    & =    & \frac{f}{m}d(S_{\mu_1(k), k}) + \frac{f}{m}d(S_{\mu_2(k),k}) +\sum_{q=1}^{f}\left(1-\frac{2}{m}\right) \max_{i, j} d_{i,j,q} \label{lem1:eq4} \\
        & \leq & \frac{f}{m}d(S_{\mu_1(k), k}) + \frac{f}{m}d(S_{\mu_2(k),k})+ \left(1-\frac{2}{m}\right) C_{k}^*. \label{lem1:eq5} 
\end{eqnarray*}
When considering the release time, coflow $k$ is transmitted starting at $\max_{k'\leq k}r_{k'}$ at the latest. This proof confirms the lemma.
\end{proof}

\begin{lem}\label{lem:lem1-2}
If $w_{k'}\geq w_{k}$, $L_{i,k'}\leq L_{i,k}$ and $L_{j,k'}\leq L_{j,k}$ hold for all $(k',k)\in E$, $i \in \mathcal{I}$ and $j \in \mathcal{J}$, then $\gamma_{k', k}=0$ holds for all $k, k'\in \mathcal{K}$.
\end{lem}
\begin{proof}
Given that $w_{k'}\geq w_{k}$, $L_{i,k'}\leq L_{i,k}$ and $L_{j,k'}\leq L_{j,k}$ hold for all $(k',k)\in E$, $i \in \mathcal{I}$, and $j \in \mathcal{J}$, the $\beta$ value of coflow $k$ is smaller than that of coflow $k'$. As a result, there is no need to order the coflow $k$ by setting $\gamma_{k',k}$.
\end{proof}

\begin{lem}\label{lem:lem2}
For every coflow $k$, $\sum_{i \in \mathcal{I}}\sum_{j \in \mathcal{J}} \alpha_{i, j, k}+\sum_{i \in \mathcal{I}}\sum_{k'\geq k}\beta_{i,k'}L_{i,k}+\sum_{j \in \mathcal{J}}\sum_{k'\geq k}\beta_{j,k'}L_{j,k}+\sum_{(k',k)\in E}\gamma_{k', k}-\sum_{(k,k')\in E}\gamma_{k, k'}= w_{k}$.
\end{lem}
\begin{proof}
A coflow $k$ is included in the permutation of Algorithm~\ref{Alg_dual} only if the constraint $\sum_{i \in \mathcal{I}}\sum_{j \in \mathcal{J}} \alpha_{i, j, k}+\sum_{i \in \mathcal{I}}\sum_{S\subseteq \mathcal{F}_{i}}\beta_{i,S}L_{i,S,k}+\sum_{j \in \mathcal{J}}\sum_{S\subseteq \mathcal{F}_{j}}\beta_{j,S}L_{j,S,k}+\sum_{(k',k)\in E}\gamma_{k', k}-\sum_{(k,k')\in E}\gamma_{k, k'}\leq w_{k}$ becomes tight for this particular coflow, resulting in $\sum_{i \in \mathcal{I}}\sum_{j \in \mathcal{J}} \alpha_{i, j, k}+\sum_{i \in \mathcal{I}}\sum_{k'\geq k}\beta_{i,k'}L_{i,k}+\sum_{j \in \mathcal{J}}\sum_{k'\geq k}\beta_{j,k'}L_{j,k}+\sum_{(k',k)\in E}\gamma_{k', k}-\sum_{(k,k')\in E}\gamma_{k, k'}= w_{k}$.
\end{proof}

\begin{lem}\label{lem:lem3}
If $w_{k'}\geq w_{k}$, $L_{i,k'}\leq L_{i,k}$ and $L_{j,k'}\leq L_{j,k}$ hold for all $(k',k)\in E$, $i \in \mathcal{I}$ and $j \in \mathcal{J}$, then the total cost of the schedule is bounded as follows.
\begin{eqnarray*}
\sum_{k}w_{k}C_{k} & \leq & \left(a+\frac{2\chi}{\kappa}\right)\sum_{k=1}^{n}\sum_{i \in \mathcal{I}}\sum_{j \in \mathcal{J}} \alpha_{i, j, k}r_{k} \\
                   &      & +2\left(a\cdot \kappa+2\chi\right)\sum_{i \in \mathcal{I}}\sum_{S\subseteq \mathcal{F}_i}\beta_{i,S}f(S) \\
									 &      & +2\left(a\cdot \kappa+2\chi\right)\sum_{j \in \mathcal{J}}\sum_{S\subseteq \mathcal{F}_j}\beta_{j,S}f(S) \\
									 &      & +\left(1-\frac{2}{m}\right)\cdot OPT.
\end{eqnarray*}
\end{lem}
\begin{proof}
By applying Lemma~\ref{lem:lem1}, we have
\begin{eqnarray*}
\sum_{k=1}^{n} w_{k}C_{k} &\leq & \sum_{k=1}^{n} w_{k}\cdot A +\left(1-\frac{2}{m}\right) \sum_{k=1}^{n} w_{k}C_{k}^*
\end{eqnarray*}
where $A=a\cdot \max_{k'\leq k}r_{k'}+\chi\frac{d(S_{\mu_1(k),k})+d(S_{\mu_2(k),k})}{m}$. We have $\sum_{k=1}^{n} w_{k} C_{k}^*=OPT$. Now we focus on the first term $\sum_{k=1}^{n} w_{k}\cdot A$. By applying Lemmas~\ref{lem:lem1-2} and \ref{lem:lem2}, we have
\begin{eqnarray*}
\sum_{k=1}^{n} w_{k}\cdot A & = & \sum_{k=1}^{n} \sum_{i \in \mathcal{I}}\sum_{j \in \mathcal{J}} \alpha_{i, j, k}\cdot A \\
                            &   & +\sum_{k=1}^{n}\sum_{i \in \mathcal{I}}\sum_{k'\geq k}\beta_{i,k'}L_{i,k}\cdot A \\
                            &   & +\sum_{k=1}^{n}\sum_{j \in \mathcal{J}}\sum_{k'\geq k}\beta_{j,k'}L_{j,k} \cdot A
\end{eqnarray*}
Let's begin by bounding $\sum_{k=1}^{n} \sum_{i \in \mathcal{I}}\sum_{j \in \mathcal{J}} \alpha_{i, j, k}\cdot A$.
By applying Observation~\ref{obs:1} parts (\ref{obs:1-5}), (\ref{obs:1-6}) and (\ref{obs:1-7}), we have
\begin{flalign*}
      & \sum_{k=1}^{n}\sum_{i \in \mathcal{I}}\sum_{j \in \mathcal{J}}\alpha_{i, j, k}\cdot A \\
\leq  & \sum_{k=1}^{n}\sum_{i \in \mathcal{I}}\sum_{j \in \mathcal{J}} \alpha_{i, j, k}\left(a\cdot r_{k}+2\chi\cdot \frac{r_{k}}{\kappa}\right) \\
\leq  & \left(a+\frac{2\chi}{\kappa}\right)\sum_{k=1}^{n}\sum_{i \in \mathcal{I}}\sum_{j \in \mathcal{J}} \alpha_{i, j, k}\cdot r_{k}
\end{flalign*}
Now we bound $\sum_{k=1}^{n}\sum_{i \in \mathcal{I}}\sum_{k'\geq k}\beta_{i,k'}L_{i,k}\cdot A$. By applying Observation~\ref{obs:1} part (\ref{obs:1-3}), we have
\begin{flalign*}
      & \sum_{k=1}^{n}\sum_{i \in \mathcal{I}}\sum_{k'\geq k}\beta_{i,k'}L_{i,k}\cdot A\\
\leq  & \sum_{k=1}^{n}\sum_{i \in \mathcal{I}}\sum_{k'\geq k}\beta_{i,k'}L_{i,k}\scriptstyle{\left(a\cdot \max_{\ell\leq k}r_{\ell}+\chi\frac{d(S_{\mu_1(k),k})+d(S_{\mu_2(k),k})}{m}\right)} \\
\leq  & \sum_{k=1}^{n}\sum_{i \in \mathcal{I}}\sum_{k'\geq k}\beta_{i,k'}L_{i,k}\scriptstyle{\left(a\cdot \kappa \cdot \frac{d(S_{\mu_1(k'),k'})}{m} + 2\chi\cdot \frac{d(S_{\mu_1(k),k})}{m}\right)} \\ 
\leq  & \left(a\cdot \kappa+2\chi\right)\sum_{k'=1}^{n}\sum_{i \in \mathcal{I}}\sum_{k\leq k'}\beta_{i,k'}L_{i,k}\frac{d(S_{\mu_1(k'),k'})}{m} \\  
\leq  & \left(a\cdot \kappa+2\chi\right)\sum_{k'=1}^{n}\sum_{i \in \mathcal{I}}\beta_{i,k'}\sum_{k\leq k'}L_{i,k}\frac{d(S_{\mu_1(k'),k'})}{m} \\
=     & \left(a\cdot \kappa+2\chi\right)\sum_{k'=1}^{n}\sum_{i \in \mathcal{I}}\beta_{i,k'}d(S_{i,k'})\frac{d(S_{\mu_1(k'),k'})}{m} \\
\leq  & \left(a\cdot \kappa+2\chi\right)\sum_{k'=1}^{n}\sum_{i \in \mathcal{I}}\beta_{i,k'}\frac{\left(d(S_{\mu_1(k'),k'})\right)^2}{m}
\end{flalign*}
By sequentially applying Observation~\ref{obs:2} and Observation~\ref{obs:1} part (\ref{obs:1-1}), we can upper bound this expression by
\begin{flalign*}
   & 2\left(a\cdot \kappa+2\chi\right)\sum_{i \in \mathcal{I}}\sum_{k=1}^{n}\beta_{i,k}f(S_{\mu_1(k),k}) \\
=  & 2\left(a\cdot \kappa+2\chi\right)\sum_{k=1}^{n}\beta_{\mu_1(k),k}f(S_{\mu_1(k),k}) \\
\leq  & 2\left(a\cdot \kappa+2\chi\right)\sum_{i \in \mathcal{I}}\sum_{S\subseteq \mathcal{F}_i}\beta_{i,S}f(S)
\end{flalign*}
By Observation~\ref{obs:2} and Observation~\ref{obs:1} parts (\ref{obs:1-2}) and (\ref{obs:1-4}), we also can obtain 
\begin{flalign*}
      & \sum_{k=1}^{n}\sum_{j \in \mathcal{J}}\sum_{k'\geq k}\beta_{j,k'}L_{j,k} \cdot A \\
\leq  & 2\left(a\cdot \kappa+2\chi\right)\sum_{j \in \mathcal{J}}\sum_{S\subseteq \mathcal{F}_j}\beta_{j,S}f(S)
\end{flalign*}
Therefore,
\begin{eqnarray*}
\sum_{k}w_{k}C_{k} & \leq & \left(a+\frac{2\chi}{\kappa}\right)\sum_{k=1}^{n}\sum_{i \in \mathcal{I}}\sum_{j \in \mathcal{J}} \alpha_{i, j, k}r_{k} \\
                   &      & +2\left(a\cdot \kappa+2\chi\right)\sum_{i \in \mathcal{I}}\sum_{S\subseteq \mathcal{F}_i}\beta_{i,S}f(S) \\
									 &      & +2\left(a\cdot \kappa+2\chi\right)\sum_{j \in \mathcal{J}}\sum_{S\subseteq \mathcal{F}_j}\beta_{j,S}f(S) \\
									 &      & +\left(1-\frac{2}{m}\right)\cdot OPT.
\end{eqnarray*}
\end{proof}

\begin{thm}\label{thm:thm1}
If $w_{k'}\geq w_{k}$, $L_{i,k'}\leq L_{i,k}$ and $L_{j,k'}\leq L_{j,k}$ hold for all $(k',k)\in E$, $i \in \mathcal{I}$ and $j \in \mathcal{J}$, then there exists a deterministic, combinatorial, polynomial time algorithm that achieves an approximation ratio of $4\chi+2-\frac{2}{m}$ for the flow-level scheduling problem with release times.
\end{thm}
\begin{proof}
To schedule coflows without release times, the application of Lemma~\ref{lem:lem3} (with $a = 1$) indicates the following:
\begin{eqnarray*}
\sum_{k}w_{k}C_{k} & \leq & \left(1+\frac{2\chi}{\kappa}\right)\sum_{k=1}^{n}\sum_{i \in \mathcal{I}}\sum_{j \in \mathcal{J}} \alpha_{i, j, k}r_{k} \\
                   &      & +2\left(\kappa+2\chi\right)\sum_{i \in \mathcal{I}}\sum_{S\subseteq \mathcal{F}_i}\beta_{i,S}f(S) \\
									 &      & +2\left(\kappa+2\chi\right)\sum_{j \in \mathcal{J}}\sum_{S\subseteq \mathcal{F}_j}\beta_{j,S}f(S) \\
									 &      & +\left(1-\frac{2}{m}\right)\cdot OPT.
\end{eqnarray*}									
In order to minimize the approximation ratio, we can substitute $\kappa=\frac{1}{2}$ and obtain the following result:
\begin{eqnarray*}
\sum_{k}w_{k}C_{k} & \leq & (4\chi+1)\sum_{k=1}^{n}\sum_{i \in \mathcal{I}}\sum_{j \in \mathcal{J}} \alpha_{i, j, k}r_{k} \\
                   &      & +(4\chi+1)\sum_{i \in \mathcal{I}}\sum_{S\subseteq \mathcal{F}_i}\beta_{i,S}f(S) \\
									 &      & +(4\chi+1)\sum_{j \in \mathcal{J}}\sum_{S\subseteq \mathcal{F}_j}\beta_{j,S}f(S) \\
									 &      & +\left(1-\frac{2}{m}\right)\cdot OPT \\
									 & \leq & \left(4\chi+2-\frac{2}{m}\right)\cdot OPT.
\end{eqnarray*}									
\end{proof}

\begin{thm}\label{thm:thm2}
If $w_{k'}\geq w_{k}$, $L_{i,k'}\leq L_{i,k}$ and $L_{j,k'}\leq L_{j,k}$ hold for all $(k',k)\in E$, $i \in \mathcal{I}$ and $j \in \mathcal{J}$, then there exists a deterministic, combinatorial, polynomial time algorithm that achieves an approximation ratio of $4\chi+1-\frac{2}{m}$ for the flow-level scheduling problem without release times.
\end{thm}
\begin{proof}
To schedule coflows without release times, the application of Lemma~\ref{lem:lem3} (with $a = 0$) indicates the following:
\begin{eqnarray*}
\sum_{k}w_{k}C_{k} & \leq & \left(\frac{2\chi}{\kappa}\right)\sum_{k=1}^{n}\sum_{i \in \mathcal{I}}\sum_{j \in \mathcal{J}} \alpha_{i, j, k}r_{k} \\
                   &      & +2\cdot 2\chi\sum_{i \in \mathcal{I}}\sum_{S\subseteq \mathcal{F}_i}\beta_{i,S}f(S) \\
									 &      & +2\cdot 2\chi\sum_{j \in \mathcal{J}}\sum_{S\subseteq \mathcal{F}_j}\beta_{j,S}f(S) \\
									 &      & +\left(1-\frac{2}{m}\right)\cdot OPT.
\end{eqnarray*}									
In order to minimize the approximation ratio, we can substitute $\kappa=\frac{1}{2}$ and obtain the following result:
\begin{eqnarray*}
\sum_{k}w_{k}C_{k} & \leq & 4\chi\sum_{k=1}^{n}\sum_{i \in \mathcal{I}}\sum_{j \in \mathcal{J}} \alpha_{i, j, k}r_{k} \\
                   &      & +4\chi\sum_{i \in \mathcal{I}}\sum_{S\subseteq \mathcal{F}_i}\beta_{i,S}f(S) \\
									 &      & +4\chi\sum_{j \in \mathcal{J}}\sum_{S\subseteq \mathcal{F}_j}\beta_{j,S}f(S) \\
									 &      & +\left(1-\frac{2}{m}\right)\cdot OPT \\
									 & \leq & \left(4\chi+1-\frac{2}{m}\right)\cdot OPT.
\end{eqnarray*}									
\end{proof}

\begin{lem}\label{lem:lem1-1}
If $L_{i,k'}\leq L_{i,k}$ and $L_{j,k'}\leq L_{j,k}$ hold for all $(k',k)\in E$, $i \in \mathcal{I}$ and $j \in \mathcal{J}$, then the inequality $\sum_{k'\in\mathcal{K}|(k',k)\in E}\gamma_{k',k}-\sum_{k'\in\mathcal{K}|(k,k')\in E}\gamma_{k,k'}\leq (R-1)(\sum_{i \in \mathcal{I}}\sum_{j \in \mathcal{J}} \alpha_{i, j, k}+\sum_{i \in \mathcal{I}}\sum_{k'\geq k}\beta_{i,k'}L_{i,k}+\sum_{j \in \mathcal{J}}\sum_{k'\geq k}\beta_{j,k'}L_{j,k})$ holds for all $k\in \mathcal{K}$.
\end{lem}
\begin{proof}
We demonstrate the case of $L_{\mu_1(r)}>L_{\mu_2(r)}$, while the other case of $L_{\mu_1(r)}\leq L_{\mu_2(r)}$ can be obtained using the same approach, yielding the same result. If coflow $k$ does not undergo the adjustment of the order by setting $\gamma_{k',k}$, then $\sum_{k'\in\mathcal{K}|(k',k)\in E}\gamma_{k',k}-\sum_{k'\in\mathcal{K}|(k,k')\in E}\gamma_{k,k'}\leq 0$.

Suppose coflow $p$ is replaced by coflow $k$ through the adjustment of $\gamma_{k',k}$.
Let
\begin{align*}
& B=\sum_{i \in \mathcal{I}}\sum_{k'\geq k}\beta_{i,k'}L_{i,k}+\sum_{j \in \mathcal{J}}\sum_{k'\geq k}\beta_{j,k'}L_{j,k}, \\
& B_{p}=\sum_{i \in \mathcal{I}}\sum_{k'\geq k}\beta_{i,k'}L_{i,p}+\sum_{j \in \mathcal{J}}\sum_{k'\geq k}\beta_{j,k'}L_{j,p}, \\
& H=\sum_{k'\in\mathcal{K}|(k',k)\in E}\gamma_{k',k}-\sum_{k'\in\mathcal{K}|(k,k')\in E}\gamma_{k,k'}, \\
& H_{p}=\sum_{k'\in\mathcal{K}|(k',p)\in E}\gamma_{k',p}-\sum_{k'\in\mathcal{K}|(p,k')\in E}\gamma_{p,k'}, \\
& R'=\frac{w_{k}}{w_{p}}.
\end{align*}
If coflow $k$ undergoes the adjustment of the order by setting $\gamma_{k',k}$, then
\begin{eqnarray}
 H & =    & w_{k}-B-\frac{L_{i,k}}{L_{i,p}}(w_{p}-B_{p}-H_{p}) \\
   & \leq & w_{k}-B-w_{p}+B_{p}+H_{p} \label{lem1-1:eq1}\\
   & \leq & w_{k}-w_{p}+H_{p} \label{lem1-1:eq2}\\
   & \leq & w_{k}-w_{p} \label{lem1-1:eq3}\\
	 & =    & \frac{R'-1}{R'}w_{k} \label{lem1-1:eq4}\\
	 & \leq & \frac{R-1}{R}w_{k} \label{lem1-1:eq5}
\end{eqnarray}
The inequalities (\ref{lem1-1:eq1}) and (\ref{lem1-1:eq2}) are due to $L_{i,p}\leq L_{i,k}$ for all $i \in \mathcal{I}$. The inequality (\ref{lem1-1:eq3}) is due to
$H_{p}\leq 0$. Based on Lemma~\ref{lem:lem2}, we know that $\sum_{i \in \mathcal{I}}\sum_{j \in \mathcal{J}} \alpha_{i, j, k}+\sum_{i \in \mathcal{I}}\sum_{k'\geq k}\beta_{i,k'}L_{i,k}+\sum_{j \in \mathcal{J}}\sum_{k'\geq k}\beta_{j,k'}L_{j,k}+\sum_{(k',k)\in E}\gamma_{k', k}-\sum_{(k,k')\in E}\gamma_{k, k'}= w_{k}$.
Thus, we obtain: 
$H \leq  (R-1)(\sum_{i \in \mathcal{I}}\sum_{j \in \mathcal{J}} \alpha_{i, j, k}+\sum_{i \in \mathcal{I}}\sum_{k'\geq k}\beta_{i,k'}L_{i,k}+\sum_{j \in \mathcal{J}}\sum_{k'\geq k}\beta_{j,k'}L_{j,k})$.
This proof confirms the lemma.
\end{proof}

\begin{lem}\label{lem:lem3-1}
If $L_{i,k'}\leq L_{i,k}$ and $L_{j,k'}\leq L_{j,k}$ hold for all $(k',k)\in E$, $i \in \mathcal{I}$ and $j \in \mathcal{J}$, then the total cost of the schedule is bounded as follows.
\begin{eqnarray*}
\sum_{k}w_{k}C_{k} & \leq & R\left(a+\frac{2\chi}{\kappa}\right)\sum_{k=1}^{n}\sum_{i \in \mathcal{I}}\sum_{j \in \mathcal{J}} \alpha_{i, j, k}r_{k} \\
                   &      & +2R\left(a\cdot \kappa+2\chi\right)\sum_{i \in \mathcal{I}}\sum_{S\subseteq \mathcal{F}_i}\beta_{i,S}f(S) \\
									 &      & +2R\left(a\cdot \kappa+2\chi\right)\sum_{j \in \mathcal{J}}\sum_{S\subseteq \mathcal{F}_j}\beta_{j,S}f(S) \\
									 &      & +\left(1-\frac{2}{m}\right)\cdot OPT.
\end{eqnarray*}
\end{lem}
\begin{proof}
According to lemma~\ref{lem:lem1-1}, we have
$\sum_{i \in \mathcal{I}}\sum_{j \in \mathcal{J}} \alpha_{i, j, k}+\sum_{i \in \mathcal{I}}\sum_{k'\geq k}\beta_{i,k'}L_{i,k}+\sum_{j \in \mathcal{J}}\sum_{k'\geq k}\beta_{j,k'}L_{j,k}+\sum_{k'\in\mathcal{K}|(k',k)\in E}\gamma_{k',k}-\sum_{k'\in\mathcal{K}|(k,k')\in E}\gamma_{k,k'}\leq R(\sum_{i \in \mathcal{I}}\sum_{j \in \mathcal{J}} \alpha_{i, j, k}+\sum_{i \in \mathcal{I}}\sum_{k'\geq k}\beta_{i,k'}L_{i,k}+\sum_{j \in \mathcal{J}}\sum_{k'\geq k}\beta_{j,k'}L_{j,k})$ holds for all $k\in \mathcal{K}$.
Then, following a similar proof to lemma~\ref{lem:lem3}, we can derive result
\begin{eqnarray*}
\sum_{k}w_{k}C_{k} & \leq & R\left(a+\frac{2\chi}{\kappa}\right)\sum_{k=1}^{n}\sum_{i \in \mathcal{I}}\sum_{j \in \mathcal{J}} \alpha_{i, j, k}r_{k} \\
                   &      & +2R\left(a\cdot \kappa+2\chi\right)\sum_{i \in \mathcal{I}}\sum_{S\subseteq \mathcal{F}_i}\beta_{i,S}f(S) \\
									 &      & +2R\left(a\cdot \kappa+2\chi\right)\sum_{j \in \mathcal{J}}\sum_{S\subseteq \mathcal{F}_j}\beta_{j,S}f(S) \\
									 &      & +\left(1-\frac{2}{m}\right)\cdot OPT.
\end{eqnarray*}
\end{proof}

By employing analogous proof techniques to theorems~\ref{thm:thm1} and \ref{thm:thm2}, we can establish the validity of the following two theorems:
\begin{thm}\label{thm:thm1-1}
If $L_{i,k'}\leq L_{i,k}$ and $L_{j,k'}\leq L_{j,k}$ hold for all $(k',k)\in E$, $i \in \mathcal{I}$ and $j \in \mathcal{J}$, then there exists a deterministic, combinatorial, polynomial time algorithm that achieves an approximation ratio of $4R\chi+R+1-\frac{2}{m}$ for the flow-level scheduling problem with release times.
\end{thm}

\begin{thm}\label{thm:thm2-1}
If $L_{i,k'}\leq L_{i,k}$ and $L_{j,k'}\leq L_{j,k}$ hold for all $(k',k)\in E$, $i \in \mathcal{I}$ and $j \in \mathcal{J}$, then there exists a deterministic, combinatorial, polynomial time algorithm that achieves an approximation ratio of $4R\chi+1-\frac{2}{m}$ for the flow-level scheduling problem without release times.
\end{thm}

\section{Approximation Algorithm for the Coflow-level Scheduling Problem}\label{sec:Algorithm2}
This section focuses on the coflow-level scheduling issue, which pertains to the transmission of flows within a coflow via a single core. It is important to remember that $L_{i,k}=\sum_{j=1}^{N}d_{i,j,k}$ and $L_{j,k}=\sum_{i=1}^{N}d_{i,j,k}$, where $L_{i,k}$ denotes the overall load at source $i$ for coflow $k$, and $L_{j,k}$ denotes the overall load at destination $j$ for coflow $k$.
Let
\begin{eqnarray*}
f_{i}(S) = \frac{\sum_{k\in S} L_{i,k}^2+\left(\sum_{k\in S} L_{i,k}\right)^2}{2m}
\end{eqnarray*}
and
\begin{eqnarray*}
f_{j}(S) = \frac{\sum_{k\in S} L_{j,k}^2+\left(\sum_{k\in S} L_{j,k}\right)^2}{2m}
\end{eqnarray*}
for any subset $S\subseteq \mathcal{K}$.
To address this problem, we propose a linear programming relaxation formulation as follows:
\begin{subequations}\label{incoflow:main}
\begin{align}
& \text{min}  && \sum_{k \in \mathcal{K}} w_{k} C_{k}     &   & \tag{\ref{incoflow:main}} \\
& \text{s.t.} && C_{k}\geq r_k+L_{i,k}, && \forall k\in \mathcal{K}, \forall i\in \mathcal{I} \label{incoflow:a} \\
&             && C_{k}\geq r_k+L_{j,k}, && \forall k\in \mathcal{K}, \forall j\in \mathcal{J} \label{incoflow:b} \\
&  && C_{k}\geq C_{k'}+L_{ik}, && \forall k, k'\in \mathcal{K}, \forall i\in \mathcal{I}: \notag \\
&  &&                          && k'\prec k\label{incoflow:prec:e} \\
&  && C_{k}\geq C_{k'}+L_{jk}, && \forall k, k'\in \mathcal{K}, \forall j\in \mathcal{J}: \notag \\
&  &&                          && k'\prec k\label{incoflow:prec:f} \\
&             && \sum_{k\in S}L_{i,k}C_{k}\geq f_{i}(S)&&  \forall i\in \mathcal{I}, \forall S\subseteq \mathcal{K} \label{incoflow:c} \\
&             && \sum_{k\in S}L_{j,k}C_{k}\geq f_{j}(S)&& \forall j\in \mathcal{J}, \forall S\subseteq \mathcal{K} \label{incoflow:d}  
\end{align}
\end{subequations}

In the linear program (\ref{incoflow:main}), the completion time $C_{k}$ is defined for each coflow $k$ in the schedule. Constraints (\ref{incoflow:a}) and (\ref{incoflow:b}) ensure that the completion time of any coflow $k$ is greater than or equal to its release time $r_k$ plus its load. To account for the precedence constraints among coflows, constraints (\ref{incoflow:prec:e}) and (\ref{incoflow:prec:f}) indicate that all flows of coflow $k'$ must be completed before coflow $k$ can be scheduled. Additionally, constraints (\ref{incoflow:c}) and (\ref{incoflow:d}) establish lower bounds for the completion time variable at the input and output ports, respectively.

The dual linear program is given by
\begin{subequations}\label{incoflow:dual}
\begin{align}
& \text{max}  && \sum_{k \in \mathcal{K}}\sum_{i \in \mathcal{I}} \alpha_{i, k}(r_k+L_{i,k}) &   &\notag\\
&   && +\sum_{k \in \mathcal{K}}\sum_{j \in \mathcal{J}} \alpha_{j, k}(r_k+L_{j,k})  &   & \notag\\
&   && +\sum_{i \in \mathcal{I}}\sum_{S \subseteq \mathcal{K}}\beta_{i,S} f_{i}(S)   &   & \notag\\
&   && +\sum_{j \in \mathcal{J}}\sum_{S \subseteq \mathcal{K}}\beta_{j,S} f_{j}(S)   &   & \notag \\
&   && + \sum_{(k', k) \in E}\sum_{i \in \mathcal{I}}\gamma_{k', i, k} L_{i,k} & & \notag \\
&   && + \sum_{(k', k) \in E}\sum_{j \in \mathcal{J}}\gamma_{k', j, k} L_{j,k} & & \tag{\ref{incoflow:dual}} \\
& \text{s.t.} && \sum_{i \in \mathcal{I}} \alpha_{i, k}+\sum_{j \in \mathcal{J}} \alpha_{j, k} &   & \notag\\
&   && +\sum_{i \in \mathcal{I}}\sum_{S\subseteq \mathcal{K}/k\in S}\beta_{i,S}L_{i,k} &   &\notag\\
&   && +\sum_{j \in \mathcal{J}}\sum_{S\subseteq \mathcal{K}/k\in S}\beta_{j,S}L_{j,k}&   &\notag\\
&   && +\sum_{(k',k)\in E}\sum_{i \in \mathcal{I}}\gamma_{k', i, k} &   &\notag\\
&   && +\sum_{(k',k)\in E}\sum_{j \in \mathcal{J}}\gamma_{k', j, k} &   &\notag\\
&   && -\sum_{(k,k')\in E}\sum_{i \in \mathcal{I}}\gamma_{k, i, k'} &   &\notag\\
&   && -\sum_{(k,k')\in E}\sum_{j \in \mathcal{J}}\gamma_{k, j, k'}\leq w_{k}, && \forall k\in \mathcal{K} \label{incoflow:dual:a} \\
&  && \alpha_{i, k} \geq 0, && \forall k\in \mathcal{K}, \forall i\in \mathcal{I} \label{incoflow:dual:b} \\
&  && \alpha_{j, k} \geq 0, && \forall k\in \mathcal{K}, \forall j\in \mathcal{J} \label{incoflow:dual:b2} \\
&  && \beta_{i, S}\geq 0,   &&  \forall i\in \mathcal{I}, \forall S\subseteq \mathcal{K} \label{incoflow:dual:c} \\
&  && \beta_{j, S}\geq 0,   &&  \forall j\in \mathcal{J}, \forall S\subseteq \mathcal{K} \label{incoflow:dual:d} \\
&  && \gamma_{k', i, k}\geq 0,   &&  \forall (k', k)\in E, \forall i\in \mathcal{I} \label{incoflow:dual:e} \\
&  && \gamma_{k', j, k}\geq 0,   &&  \forall (k', k)\in E, \forall j\in \mathcal{J}\label{incoflow:dual:f}
\end{align}
\end{subequations}

Let $\gamma_{k', k}=\sum_{i \in \mathcal{I}}\gamma_{k', i, k}+\sum_{j \in \mathcal{J}}\gamma_{k', j, k}$. Notice that for every coflow $k$, there exists two dual variables $\alpha_{i, k}$ and $\alpha_{j, k}$, and there is a corresponding constraint. Additionally, for every subset of coflows $S$, there are two dual variables $\beta_{i, S}$ and $\beta_{j, S}$. For the precedence constraints, there are two dual variables $\gamma_{k', k}$ and $\gamma_{k, k'}$. Algorithm~\ref{Alg2_dual} in Appendix~\ref{appendix:b} presents the primal-dual algorithm which has a space complexity of $O(Nn)$ and a time complexity of $O(n^2)$, where $N$ represents the number of input/output ports and $n$ represents the number of coflows.

The coflow-driven-list-scheduling, as outlined in Algorithm~\ref{Alg2}, operates as follows. To ensure clarity and generality, we assume that the coflows are arranged in an order determined by the permutation generated by Algorithm~\ref{Alg2_dual}, where $\sigma(k)=k$ for all $k\in \mathcal{K}$. We schedule all the flows within each coflow iteratively, following the sequence provided by this list.

For each coflow $k$, we identify the network core $h^*$ that can transmit coflow $k$ in a manner that minimizes its completion time (lines \ref{alg2-2}-\ref{alg2-3}). Subsequently, we transmit all the flows allocated to network core $h$ (lines \ref{alg2-4}-\ref{alg2-5}).

In summary, the coflow-driven-list-scheduling algorithm works by iteratively scheduling the flows within each coflow, following a predetermined order. It determines the optimal network core for transmitting each coflow to minimize their completion times, and then transmits the allocated flows for each core accordingly.

\begin{algorithm}
\caption{coflow-driven-list-scheduling}
    \begin{algorithmic}[1]
				\STATE Both $load_{I}$ and $load_{O}$ are initialized to zero and $\mathcal{A}_h=\emptyset$ for all $h\in [1, m]$
				\FOR{$k=1, 2, \ldots, n$} \label{alg2-2}
				    \STATE $h^*=\arg \min_{h\in [1, m]}\left(\max_{i,j\in [1,N]}load_{I}(i,h)+\right.$ $\left.load_{O}(j,h)+L_{i,k}+L_{j,k}\right)$
						\STATE $\mathcal{A}_{h^*}=\mathcal{A}_{h^*}\cup \left\{k\right\}$
						\STATE $load_{I}(i,h^*)=load_{I}(i,h^*)+L_{i,k}$ and $load_{O}(j,h^*)=load_{O}(j,h^*)+L_{j,k}$ for all $i,j\in [1,N]$
				\ENDFOR               \label{alg2-3}
				\FOR{each $h\in [1, m]$ do in parallel}\label{alg2-4}
				    \STATE wait until the first coflow is released
						\WHILE{there is some incomplete flow}
						    \STATE for all $k\in \mathcal{A}_{h}$, list the released and incomplete flows respecting the increasing order in $k$
								\STATE let $L$ be the set of flows in the list
                \FOR{every flow $f=(i, j, k)\in L$}
										\IF{the link $(i, j)$ is idle}
										    \STATE schedule flow $f$ \label{alg2-1}
										\ENDIF
								\ENDFOR
								\WHILE{no new flow is completed or released}
								    \STATE transmit the flows that get scheduled in line \ref{alg2-1} at maximum rate 1.
								\ENDWHILE
						\ENDWHILE
				\ENDFOR\label{alg2-5}
   \end{algorithmic}
\label{Alg2}
\end{algorithm}

\subsection{Analysis}
In this section, we present a comprehensive analysis of the proposed algorithm, establishing its approximation ratios. Specifically, we demonstrate that the algorithm achieves an approximation ratio of $O(m\chi)$ when considering workload sizes and weights that are topology-dependent in the input instances. Additionally, when considering workload sizes that are topology-dependent in the input instances, the algorithm achieves an approximation ratio of $O(Rm\chi)$ where $R$ is the ratio of maximum weight to minimum weight. It is crucial to note that our analysis assumes that the coflows are arranged in the order determined by the permutation generated by Algorithm~\ref{Alg2_dual}, where $\sigma(k)=k$ for all $k\in \mathcal{K}$. 

We would like to emphasize that $S_{k}=\left\{1, 2, \ldots, k\right\}$ represents the set of the first $k$ coflows. We define $\beta_{i,k}=\beta_{i,S_{k}}$ and $\beta_{j,k}=\beta_{j,S_{k}}$ for convenience. Moreover, we define $L_{i}(S_{k})=\sum_{k'\leq k} L_{i, k'}$ and $L_{j}(S_{k})=\sum_{k'\leq k} L_{j, k'}$ to simplify the notation. Furthermore, let $\mu_1(k)$ denote the input port with the highest load among the coflows in $S_{k}$, and $\mu_2(k)$ denote the output port with the highest load among the coflows in $S_{k}$. Hence, we have $L_{\mu_1(k)}(S_{k})=\sum_{k'\leq k} L_{\mu_1(k), k'}$ and $L_{\mu_2(k)}(S_{k})=\sum_{k'\leq k} L_{\mu_2(k), k'}$.

Let us begin by presenting several key observations regarding the primal-dual algorithm.
\begin{obs}\label{obs:3}
The following statements hold.

\begin{enumerate}
\item Every nonzero $\beta_{i,S}$ can be written as $\beta_{\mu_1(k),k}$ for some coflow $k$. \label{obs:3-1}
\item Every nonzero $\beta_{j,S}$ can be written as $\beta_{\mu_2(k),k}$ for some coflow $k$. \label{obs:3-2}
\item For every set $S_{k}$ that has a nonzero $\beta_{\mu_1(k),k}$ variable, if $k' \leq k$ then $r_{k'}\leq \frac{\kappa\cdot L_{\mu_1(k)}(S_{k})}{m}$. \label{obs:3-3}
\item For every set $S_{k}$ that has a nonzero $\beta_{\mu_2(k),k}$ variable, if $k' \leq k$ then $r_{k'}\leq \frac{\kappa\cdot L_{\mu_2(k)}(S_{k})}{m}$. \label{obs:3-4}
\item For every coflow $k$ that has a nonzero $\alpha_{\mu_1(k), k}$, $r_{k}>\frac{\kappa\cdot L_{\mu_1(k)}(S_{k})}{m}$. \label{obs:3-5}
\item For every coflow $k$ that has a nonzero $\alpha_{\mu_2(k), k}$, $r_{k}>\frac{\kappa\cdot L_{\mu_2(k)}(S_{k})}{m}$. \label{obs:3-6}
\item For every coflow $k$ that has a nonzero $\alpha_{\mu_1(k), k}$ or a nonzero $\alpha_{\mu_2(k), k}$, if $k'\leq k$ then $r_{k'}\leq r_{k}$. \label{obs:3-7}
\end{enumerate}
\end{obs}
The validity of each of the aforementioned observations can be readily verified and directly inferred from the steps outlined in Algorithm~\ref{Alg2_dual}. 
\begin{obs}\label{obs:4}
For any subset $S$, we have that $(\sum_{k\in S} L_{i,k})^2\leq 2m\cdot f_{i}(S)$ and $(\sum_{k\in S} L_{j,k})^2\leq 2m\cdot f_{j}(S)$. 
\end{obs} 

\begin{lem}\label{lem:lem21}
Let $C_{k}$ represent the completion time of coflow $k$ when scheduled according to Algorithm~\ref{Alg2}. For any coflow $k$, we have $C_{k}\leq a\cdot \max_{k'\leq k}r_{k'}+\chi\left(L_{\mu_1(k)}(S_{k})+L_{\mu_2(k)}(S_{k})\right)$, where $a=0$ signifies the absence of release times, and $a=1$ indicates the presence of arbitrary release times.
\end{lem}
\begin{proof}
First, let's consider the case where there is no release time and no precedence constraints. In this case, the completion time bound for each coflow can be expressed by the following inequality:
\begin{eqnarray*}
\hat{C}_{k} & \leq & L_{\mu_1(k)}(S_{k})+L_{\mu_2(k)}(S_{k}) \label{lem21:eq1}
\end{eqnarray*}

Now, let $v_1v_2\cdots v_f$ be the longest path of coflow $k$, where $v_f=k$. Then, we can derive the following inequalities:
\begin{eqnarray*}
C_{k}   & \leq & \sum_{q=1}^{f} \hat{C}_{v_q} \label{lem21:eq1}\\
        & \leq & \sum_{q=1}^{f} L_{\mu_1(q)}(S_{q})+L_{\mu_2(q)}(S_{q})  \label{lem21:eq2}\\
		    & \leq & \sum_{q=1}^{f} L_{\mu_1(k)}(S_{k})+L_{\mu_2(k)}(S_{k})\label{lem21:eq3}\\
		    & =    & f\left(L_{\mu_1(k)}(S_{k})+L_{\mu_2(k)}(S_{k})\right) \label{lem21:eq4} 
\end{eqnarray*}
When considering the release time, coflow $k$ is transmitted starting at $\max_{k'\leq k}r_{k'}$ at the latest. This proof confirms the lemma.
\end{proof}

\begin{lem}\label{lem:lem2-2}
If $w_{k'}\geq w_{k}$, $L_{i,k'}\leq L_{i,k}$ and $L_{j,k'}\leq L_{j,k}$ hold for all $(k',k)\in E$, $i \in \mathcal{I}$ and $j \in \mathcal{J}$, then $\gamma_{k', k}=0$ holds for all $k, k'\in \mathcal{K}$.
\end{lem}
\begin{proof}
Given that $w_{k'}\geq w_{k}$, $L_{i,k'}\leq L_{i,k}$ and $L_{j,k'}\leq L_{j,k}$ hold for all $(k',k)\in E$, $i \in \mathcal{I}$, and $j \in \mathcal{J}$, the $\beta$ value of coflow $k$ is smaller than that of coflow $k'$. As a result, there is no need to order the coflow $k$ by setting $\gamma_{k',k}$.
\end{proof}

\begin{lem}\label{lem:lem22}
For every coflow $k$, $\sum_{i \in \mathcal{I}} \alpha_{i, k}+\sum_{j \in \mathcal{J}} \alpha_{j, k}+\sum_{i \in \mathcal{I}}\sum_{k'\geq k}\beta_{i,k'}L_{i,k}+\sum_{j \in \mathcal{J}}\sum_{k'\geq k}\beta_{j,k'}L_{j,k}+\sum_{k'\in\mathcal{K}|(k',k)\in E}\gamma_{k',k}-\sum_{k'\in\mathcal{K}|(k,k')\in E}\gamma_{k,k'}= w_{k}$.
\end{lem}
\begin{proof}
A coflow $k$ is included in the permutation of Algorithm~\ref{Alg2_dual} only if the constraint 
$\sum_{i \in \mathcal{I}} \alpha_{i, k}+\sum_{j \in \mathcal{J}} \alpha_{j, k} +\sum_{i \in \mathcal{I}}\sum_{S\subseteq \mathcal{K}/k\in S}\beta_{i,S}L_{i,k} +\sum_{j \in \mathcal{J}}\sum_{S\subseteq \mathcal{K}/k\in S}\beta_{j,S}L_{j,k}+\sum_{k'\in\mathcal{K}|(k',k)\in E}\gamma_{k',k}-\sum_{k'\in\mathcal{K}|(k,k')\in E}\gamma_{k,k'}\leq w_{k}$ becomes tight for this particular coflow, resulting in $\sum_{i \in \mathcal{I}} \alpha_{i, k}+\sum_{j \in \mathcal{J}} \alpha_{j, k}+\sum_{i \in \mathcal{I}}\sum_{k'\geq k}\beta_{i,k'}L_{i,k}+\sum_{j \in \mathcal{J}}\sum_{k'\geq k}\beta_{j,k'}L_{j,k}+\sum_{k'\in\mathcal{K}|(k',k)\in E}\gamma_{k',k}-\sum_{k'\in\mathcal{K}|(k,k')\in E}\gamma_{k,k'}= w_{k}$.
\end{proof}

\begin{lem}\label{lem:lem23}
If $w_{k'}\geq w_{k}$, $L_{i,k'}\leq L_{i,k}$ and $L_{j,k'}\leq L_{j,k}$ hold for all $(k',k)\in E$, $i \in \mathcal{I}$ and $j \in \mathcal{J}$, then the total cost of the schedule is bounded as follows.
\begin{eqnarray*}
\sum_{k}w_{k}C_{k} & \leq & \left(a+\frac{2\chi\cdot m}{\kappa}\right)\sum_{k \in \mathcal{K}}\sum_{i \in \mathcal{I}} \alpha_{i, k}(r_{k}) \\
                   &      & +\left(a+\frac{2\chi\cdot m}{\kappa}\right)\sum_{k \in \mathcal{K}}\sum_{j \in \mathcal{J}} \alpha_{j, k}(r_k) \\
                   &      & +2\left(a\cdot \kappa+2\chi\cdot m\right)\sum_{i \in \mathcal{I}}\sum_{S \subseteq \mathcal{K}}\beta_{i,S} f_{i}(S) \\
									 &      & +2\left(a\cdot \kappa+2\chi\cdot m\right)\sum_{j \in \mathcal{J}}\sum_{S \subseteq \mathcal{K}}\beta_{j,S} f_{j}(S) 
\end{eqnarray*}
\end{lem}
\begin{proof}
By applying Lemma~\ref{lem:lem21}, we have
\begin{eqnarray*}
&& \sum_{k=1}^{n} w_{k}C_{k} \\
&& \leq \sum_{k=1}^{n} w_{k}\cdot \left(a\cdot \max_{k'\leq k}r_k+\chi\left(L_{\mu_1(k)}(S_{k})+L_{\mu_2(k)}(S_{k})\right)\right)
\end{eqnarray*}
Let $A=a\cdot \max_{k'\leq k}r_k+\chi\left(L_{\mu_1(k)}(S_{k})+L_{\mu_2(k)}(S_{k})\right)$. By applying Lemmas~\ref{lem:lem2-2} and \ref{lem:lem22}, we have
\begin{eqnarray*}
\sum_{k=1}^{n} w_{k}C_{k}   & \leq & \sum_{k=1}^{n}\left(\sum_{i \in \mathcal{I}} \alpha_{i, k}+\sum_{k=1}^{n}\sum_{j \in \mathcal{J}} \alpha_{j, k}\right)\cdot A\\
                            &      & +\sum_{k=1}^{n}\sum_{i \in \mathcal{I}}\sum_{k'\geq k}\beta_{i,k'}L_{i,k} \cdot A\\
														&      & +\sum_{k=1}^{n}\sum_{j \in \mathcal{J}}\sum_{k'\geq k}\beta_{j,k'}L_{j,k} \cdot A\\
\end{eqnarray*}
Let's begin by bounding $\sum_{k=1}^{n}\sum_{i \in \mathcal{I}} \alpha_{i, k} \cdot A+\sum_{k=1}^{n}\sum_{j \in \mathcal{J}} \alpha_{j, k} \cdot A$.
By applying Observation~\ref{obs:3} parts (\ref{obs:3-5}), (\ref{obs:3-6}) and (\ref{obs:3-7}), we have
\begin{flalign*}
      & \sum_{k=1}^{n}\left(\sum_{i \in \mathcal{I}} \alpha_{i, k}+\sum_{k=1}^{n}\sum_{j \in \mathcal{J}} \alpha_{j, k}\right)\cdot A \\
\leq  & \sum_{k=1}^{n}\left(\sum_{i \in \mathcal{I}} \alpha_{i, k}+\sum_{k=1}^{n}\sum_{j \in \mathcal{J}} \alpha_{j, k}\right)\left(a\cdot r_{k}+2\chi\cdot m\cdot \frac{r_{k}}{\kappa}\right) \\
\leq  & \left(a+\frac{2\chi\cdot m}{\kappa}\right)\sum_{k=1}^{n}\left(\sum_{i \in \mathcal{I}} \alpha_{i, k}+\sum_{k=1}^{n}\sum_{j \in \mathcal{J}} \alpha_{j, k}\right)\cdot r_{k}
\end{flalign*}
Now we bound $\sum_{k=1}^{n}\sum_{i \in \mathcal{I}}\sum_{k'\geq k}\beta_{i,k'}L_{i,k} \cdot A$. By applying Observation~\ref{obs:3} part (\ref{obs:3-3}), we have
\begin{flalign*}
      & \sum_{k=1}^{n}\sum_{i \in \mathcal{I}}\sum_{k'\geq k}\beta_{i,k'}L_{i,k} \cdot A\\
\leq  & \sum_{k=1}^{n}\sum_{i \in \mathcal{I}}\sum_{k'\geq k}\beta_{i,k'}L_{i,k}\left(a\cdot \max_{k'\leq k}r_{k'}+L_{\mu_1(k)}(S_{k})+L_{\mu_2(k)}(S_{k})\right) \\
\leq  & \sum_{k=1}^{n}\sum_{i \in \mathcal{I}}\sum_{k'\geq k}\beta_{i,k'}L_{i,k}\left(a\cdot \kappa \cdot \frac{L_{\mu_1(k)}(S_{k})}{m} + 2\chi\cdot L_{\mu_1(k)}(S_{k})\right) \\ 
\leq  & \left(a\cdot \kappa+2\chi\cdot m\right)\sum_{k'=1}^{n}\sum_{i \in \mathcal{I}}\sum_{k\leq k'}\beta_{i,k'}L_{i,k}\frac{L_{\mu_1(k)}(S_{k})}{m} \\  
\leq  & \left(a\cdot \kappa+2\chi\cdot m\right)\sum_{k'=1}^{n}\sum_{i \in \mathcal{I}}\beta_{i,k'}\sum_{k\leq k'}L_{i,k}\frac{L_{\mu_1(k)}(S_{k})}{m} \\
=     & \left(a\cdot \kappa+2\chi\cdot m\right)\sum_{k'=1}^{n}\sum_{i \in \mathcal{I}}\beta_{i,k'}L_{i}(S_{k})\frac{L_{\mu_1(k)}(S_{k})}{m} \\
\leq  & \left(a\cdot \kappa+2\chi\cdot m\right)\sum_{k'=1}^{n}\sum_{i \in \mathcal{I}}\beta_{i,k'}\frac{\left(L_{\mu_1(k)}(S_{k})\right)^2}{m}
\end{flalign*}
By sequentially applying Observation~\ref{obs:4} and Observation~\ref{obs:3} part (\ref{obs:3-1}), we can upper bound this expression by
\begin{flalign*}
   & 2\left(a\cdot \kappa+2\chi\cdot m\right)\sum_{i \in \mathcal{I}}\sum_{k=1}^{n}\beta_{i,k}f_{i}(S_{\mu_1(k),k}) \\
=  & 2\left(a\cdot \kappa+2\chi\cdot m\right)\sum_{k=1}^{n}\beta_{\mu_1(k),k}f_{i}(S_{\mu_1(k),k}) \\
\leq  & 2\left(a\cdot \kappa+2\chi\cdot m\right)\sum_{i \in \mathcal{I}}\sum_{S\subseteq \mathcal{K}}\beta_{i,S}f_{i}(S)
\end{flalign*}
By Observation~\ref{obs:4} and Observation~\ref{obs:3} parts (\ref{obs:3-2}) and (\ref{obs:3-4}), we also can obtain 
\begin{flalign*}
      & \sum_{k=1}^{n}\sum_{j \in \mathcal{J}}\sum_{k'\geq k}\beta_{j,k'}L_{j,k} \cdot A \\
\leq  & 2\left(a\cdot \kappa+2\chi\cdot m\right)\sum_{j \in \mathcal{J}}\sum_{S\subseteq \mathcal{K}}\beta_{j,S}f_{j}(S)
\end{flalign*}
Therefore,
\begin{eqnarray*}
\sum_{k}w_{k}C_{k} & \leq & \left(a+\frac{2\chi\cdot m}{\kappa}\right)\sum_{k \in \mathcal{K}}\sum_{i \in \mathcal{I}} \alpha_{i, k}(r_{k}) \\
                   &      & +\left(a+\frac{2\chi\cdot m}{\kappa}\right)\sum_{k \in \mathcal{K}}\sum_{j \in \mathcal{J}} \alpha_{j, k}(r_k) \\
                   &      & +2\left(a\cdot \kappa+2\chi\cdot m\right)\sum_{i \in \mathcal{I}}\sum_{S \subseteq \mathcal{K}}\beta_{i,S} f_{i}(S) \\
									 &      & +2\left(a\cdot \kappa+2\chi\cdot m\right)\sum_{j \in \mathcal{J}}\sum_{S \subseteq \mathcal{K}}\beta_{j,S} f_{j}(S) 
\end{eqnarray*}
\end{proof}

\begin{thm}\label{thm:thm21}
If $w_{k'}\geq w_{k}$, $L_{i,k'}\leq L_{i,k}$ and $L_{j,k'}\leq L_{j,k}$ hold for all $(k',k)\in E$, $i \in \mathcal{I}$ and $j \in \mathcal{J}$, there exists a deterministic, combinatorial, polynomial time algorithm that achieves an approximation ratio of $4\chi m+1$ for the coflow-level scheduling problem with release times.
\end{thm}
\begin{proof}
To schedule coflows without release times, the application of Lemma~\ref{lem:lem23} (with $a = 1$) indicates the following:
\begin{eqnarray*}
\sum_{k}w_{k}C_{k} & \leq & \left(1+\frac{2\chi\cdot m}{\kappa}\right)\sum_{k \in \mathcal{K}}\sum_{i \in \mathcal{I}} \alpha_{i, k}(r_{k}) \\
                   &      & +\left(1+\frac{2\chi\cdot m}{\kappa}\right)\sum_{k \in \mathcal{K}}\sum_{j \in \mathcal{J}} \alpha_{j, k}(r_k) \\
                   &      & +2\left(\kappa+2\chi\cdot m\right)\sum_{i \in \mathcal{I}}\sum_{S \subseteq \mathcal{K}}\beta_{i,S} f_{i}(S) \\
									 &      & +2\left(\kappa+2\chi\cdot m\right)\sum_{j \in \mathcal{J}}\sum_{S \subseteq \mathcal{K}}\beta_{j,S} f_{j}(S) 
\end{eqnarray*}
In order to minimize the approximation ratio, we can substitute $\kappa=\frac{1}{2}$ and obtain the following result:
\begin{eqnarray*}
\sum_{k}w_{k}C_{k} & \leq & \left(4\chi\cdot m+1\right)\sum_{k \in \mathcal{K}}\sum_{i \in \mathcal{I}} \alpha_{i, k}(r_{k}) \\
                   &      & +\left(4\chi\cdot m+1\right)\sum_{k \in \mathcal{K}}\sum_{j \in \mathcal{J}} \alpha_{j, k}(r_k) \\
                   &      & +\left(4\chi\cdot m+1\right)\sum_{i \in \mathcal{I}}\sum_{S \subseteq \mathcal{K}}\beta_{i,S} f_{i}(S) \\
									 &      & +\left(4\chi\cdot m+1\right)\sum_{j \in \mathcal{J}}\sum_{S \subseteq \mathcal{K}}\beta_{j,S} f_{j}(S) \\
									 & \leq & \left(4\chi\cdot m+1\right) \cdot OPT.
\end{eqnarray*}
\end{proof}

\begin{thm}\label{thm:thm22}
If $w_{k'}\geq w_{k}$, $L_{i,k'}\leq L_{i,k}$ and $L_{j,k'}\leq L_{j,k}$ hold for all $(k',k)\in E$, $i \in \mathcal{I}$ and $j \in \mathcal{J}$, there exists a deterministic, combinatorial, polynomial time algorithm that achieves an approximation ratio of $4\chi m$ for the coflow-level scheduling problem without release times.
\end{thm}
\begin{proof}
To schedule coflows without release times, the application of Lemma~\ref{lem:lem23} (with $a = 0$) indicates the following:
\begin{eqnarray*}
\sum_{k}w_{k}C_{k} & \leq & \left(\frac{2\chi\cdot m}{\kappa}\right)\sum_{k \in \mathcal{K}}\sum_{i \in \mathcal{I}} \alpha_{i, k}(r_{k}) \\
                   &      & +\left(\frac{2\chi\cdot m}{\kappa}\right)\sum_{k \in \mathcal{K}}\sum_{j \in \mathcal{J}} \alpha_{j, k}(r_k) \\
                   &      & +2\left(2\chi\cdot m\right)\sum_{i \in \mathcal{I}}\sum_{S \subseteq \mathcal{K}}\beta_{i,S} f_{i}(S) \\
									 &      & +2\left(2\chi\cdot m\right)\sum_{j \in \mathcal{J}}\sum_{S \subseteq \mathcal{K}}\beta_{j,S} f_{j}(S) 
\end{eqnarray*}
In order to minimize the approximation ratio, we can substitute $\kappa=\frac{1}{2}$ and obtain the following result:
\begin{eqnarray*}
\sum_{k}w_{k}C_{k} & \leq & \left(4\chi\cdot m\right)\sum_{k \in \mathcal{K}}\sum_{i \in \mathcal{I}} \alpha_{i, k}(r_{k}) \\
                   &      & +\left(4\chi\cdot m\right)\sum_{k \in \mathcal{K}}\sum_{j \in \mathcal{J}} \alpha_{j, k}(r_k) \\
                   &      & +\left(4\chi\cdot m\right)\sum_{i \in \mathcal{I}}\sum_{S \subseteq \mathcal{K}}\beta_{i,S} f_{i}(S) \\
									 &      & +\left(4\chi\cdot m\right)\sum_{j \in \mathcal{J}}\sum_{S \subseteq \mathcal{K}}\beta_{j,S} f_{j}(S) \\
									 & \leq & 4\chi\cdot m \cdot OPT.
\end{eqnarray*}
\end{proof}

\begin{lem}\label{lem:lem21-1}
If $L_{i,k'}\leq L_{i,k}$ and $L_{j,k'}\leq L_{j,k}$ hold for all $(k',k)\in E$, $i \in \mathcal{I}$ and $j \in \mathcal{J}$, then the inequality $\sum_{k'\in\mathcal{K}|(k',k)\in E}\gamma_{k',k}-\sum_{k'\in\mathcal{K}|(k,k')\in E}\gamma_{k,k'}\leq (R-1)(\sum_{i \in \mathcal{I}} \alpha_{i, k}+\sum_{j \in \mathcal{J}} \alpha_{j, k}+\sum_{i \in \mathcal{I}}\sum_{k'\geq k}\beta_{i,k'}L_{i,k}+\sum_{j \in \mathcal{J}}\sum_{k'\geq k}\beta_{j,k'}L_{j,k})$ holds for all $k\in \mathcal{K}$.
\end{lem}
\begin{proof}
If coflow $k$ does not undergo the adjustment of the order by setting $\gamma_{k',k}$, then $\sum_{k'\in\mathcal{K}|(k',k)\in E}\gamma_{k',k}-\sum_{k'\in\mathcal{K}|(k,k')\in E}\gamma_{k,k'}\leq 0$. If coflow $k$ undergoes the adjustment of the order by setting $\gamma_{k',k}$, then we have $\sum_{k'\in\mathcal{K}|(k',k)\in E}\gamma_{k',k}-\sum_{k'\in\mathcal{K}|(k,k')\in E}\gamma_{k,k'}\leq \frac{R-1}{R}w_{k}$. Based on Lemma~\ref{lem:lem22}, we know that $\sum_{i \in \mathcal{I}} \alpha_{i, k}+\sum_{j \in \mathcal{J}} \alpha_{j, k}+\sum_{i \in \mathcal{I}}\sum_{k'\geq k}\beta_{i,k'}L_{i,k}+\sum_{j \in \mathcal{J}}\sum_{k'\geq k}\beta_{j,k'}L_{j,k}+\sum_{(k',k)\in E}\gamma_{k', k}-\sum_{(k,k')\in E}\gamma_{k, k'}= w_{k}$.
Thus, we obtain: 
$\sum_{k'\in\mathcal{K}|(k',k)\in E}\gamma_{k',k}-\sum_{k'\in\mathcal{K}|(k,k')\in E}\gamma_{k,k'} \leq (R-1)(\sum_{i \in \mathcal{I}} \alpha_{i, k}+\sum_{j \in \mathcal{J}} \alpha_{j, k}+\sum_{i \in \mathcal{I}}\sum_{k'\geq k}\beta_{i,k'}L_{i,k}+\sum_{j \in \mathcal{J}}\sum_{k'\geq k}\beta_{j,k'}L_{j,k})$.
This proof confirms the lemma.
\end{proof}

\begin{lem}\label{lem:lem23-1}
If $L_{i,k'}\leq L_{i,k}$ and $L_{j,k'}\leq L_{j,k}$ hold for all $(k',k)\in E$, $i \in \mathcal{I}$ and $j \in \mathcal{J}$, then the total cost of the schedule is bounded as follows.
\begin{eqnarray*}
\sum_{k}w_{k}C_{k} & \leq & R\left(a+\frac{2\chi\cdot m}{\kappa}\right)\sum_{k \in \mathcal{K}}\sum_{i \in \mathcal{I}} \alpha_{i, k}(r_{k}) \\
                   &      & +R\left(a+\frac{2\chi\cdot m}{\kappa}\right)\sum_{k \in \mathcal{K}}\sum_{j \in \mathcal{J}} \alpha_{j, k}(r_k) \\
                   &      & +2R\left(a\cdot \kappa+2\chi\cdot m\right)\sum_{i \in \mathcal{I}}\sum_{S \subseteq \mathcal{K}}\beta_{i,S} f_{i}(S) \\
									 &      & +2R\left(a\cdot \kappa+2\chi\cdot m\right)\sum_{j \in \mathcal{J}}\sum_{S \subseteq \mathcal{K}}\beta_{j,S} f_{j}(S) 
\end{eqnarray*}
\end{lem}
\begin{proof}
According to lemma~\ref{lem:lem21-1}, we have
$\sum_{i \in \mathcal{I}} \alpha_{i, k}+\sum_{j \in \mathcal{J}} \alpha_{j, k}+\sum_{i \in \mathcal{I}}\sum_{k'\geq k}\beta_{i,k'}L_{i,k}+\sum_{j \in \mathcal{J}}\sum_{k'\geq k}\beta_{j,k'}L_{j,k}+\sum_{k'\in\mathcal{K}|(k',k)\in E}\gamma_{k',k}-\sum_{k'\in\mathcal{K}|(k,k')\in E}\gamma_{k,k'}\leq R(\sum_{i \in \mathcal{I}} \alpha_{i, k}+\sum_{j \in \mathcal{J}} \alpha_{j, k}+\sum_{i \in \mathcal{I}}\sum_{k'\geq k}\beta_{i,k'}L_{i,k}+\sum_{j \in \mathcal{J}}\sum_{k'\geq k}\beta_{j,k'}L_{j,k})$ holds for all $k\in \mathcal{K}$.
Then, following a similar proof to lemma~\ref{lem:lem23}, we can derive result
\begin{eqnarray*}
\sum_{k}w_{k}C_{k} & \leq & R\left(a+\frac{2\chi\cdot m}{\kappa}\right)\sum_{k \in \mathcal{K}}\sum_{i \in \mathcal{I}} \alpha_{i, k}(r_{k}) \\
                   &      & +R\left(a+\frac{2\chi\cdot m}{\kappa}\right)\sum_{k \in \mathcal{K}}\sum_{j \in \mathcal{J}} \alpha_{j, k}(r_k) \\
                   &      & +2R\left(a\cdot \kappa+2\chi\cdot m\right)\sum_{i \in \mathcal{I}}\sum_{S \subseteq \mathcal{K}}\beta_{i,S} f_{i}(S) \\
									 &      & +2R\left(a\cdot \kappa+2\chi\cdot m\right)\sum_{j \in \mathcal{J}}\sum_{S \subseteq \mathcal{K}}\beta_{j,S} f_{j}(S) 
\end{eqnarray*}
\end{proof}

By employing analogous proof techniques to theorems~\ref{thm:thm21} and \ref{thm:thm22}, we can establish the validity of the following two theorems:
\begin{thm}\label{thm:thm21-1}
If $L_{i,k'}\leq L_{i,k}$ and $L_{j,k'}\leq L_{j,k}$ hold for all $(k',k)\in E$, $i \in \mathcal{I}$ and $j \in \mathcal{J}$, then there exists a deterministic, combinatorial, polynomial time algorithm that achieves an approximation ratio of $4R\chi m+R$ for the flow-level scheduling problem with release times.
\end{thm}

\begin{thm}\label{thm:thm22-1}
If $L_{i,k'}\leq L_{i,k}$ and $L_{j,k'}\leq L_{j,k}$ hold for all $(k',k)\in E$, $i \in \mathcal{I}$ and $j \in \mathcal{J}$, then there exists a deterministic, combinatorial, polynomial time algorithm that achieves an approximation ratio of $4R\chi m$ for the flow-level scheduling problem without release times.
\end{thm}


\section{Coflows of Multi-stage Jobs Scheduling Problem}\label{sec:Algorithm6}
In this section, we will focus on addressing the coflows of multi-stage job scheduling problem. We will modify the linear programs~(\ref{coflow:main}) by introducing a set $\mathcal{T}$ to represent the jobs and a set $\mathcal{T}_{t}$ to represent the coflows that belong to job $t$. We will also incorporate an additional constraint~(\ref{job:coflow:b}), which will ensure that the completion time of any job is limited by its coflows. Our objective is to minimize the total weighted completion time for a given set of multi-stage jobs. Assuming that all coflows within the same job have the same release time. The resulting problem can be expressed as a linear programming relaxation, which is as follows:
\begin{subequations}\label{job:coflow:main}
\begin{align}
& \text{min}  && \sum_{t \in \mathcal{T}} w_{t} C_{t}     &   & \tag{\ref{job:coflow:main}} \\
& \text{s.t.} && (\ref{coflow:a})-(\ref{coflow:d}) &&  \notag \\
&  && C_{t} \geq C_{k},&& \forall t\in \mathcal{T}, \forall k\in \mathcal{T}_{t} \label{job:coflow:b} 
\end{align}
\end{subequations}

The dual linear program is given by
\begin{subequations}\label{job:coflow:dual}
\begin{align}
& \text{max}  && \sum_{k \in \mathcal{K}}\sum_{i \in \mathcal{I}}\sum_{j \in \mathcal{J}} \alpha_{i, j, k}(r_k+d_{i,j,k}) &   &\notag\\
&   && +\sum_{i \in \mathcal{I}}\sum_{S \subseteq \mathcal{F}_{i}}\beta_{i,S} f(S)     &   & \notag\\
&   && +\sum_{j \in \mathcal{J}}\sum_{S \subseteq \mathcal{F}_{j}}\beta_{j,S} f(S)     &   & \notag \\
&   && + \sum_{(k', k) \in E}\sum_{i \in \mathcal{I},j \in \mathcal{J}}\gamma_{k', i, j, k} d_{i,j,k}     &   & \tag{\ref{job:coflow:dual}} \\
& \text{s.t.} && \sum_{k\in \mathcal{T}_{t}}\sum_{i \in \mathcal{I}}\sum_{j \in \mathcal{J}} \alpha_{i, j, k} &   & \notag\\
&   && +\sum_{k\in \mathcal{T}_{t}}\sum_{i \in \mathcal{I}}\sum_{S\subseteq \mathcal{F}_{i}}\beta_{i,S}L_{i,S,k} &   &\notag\\
&   && +\sum_{k\in \mathcal{T}_{t}}\sum_{j \in \mathcal{J}}\sum_{S\subseteq \mathcal{F}_{j}}\beta_{j,S}L_{j,S,k} &   &\notag\\
&   && +\sum_{k\in \mathcal{T}_{t}}\sum_{(k',k)\in E}\sum_{i \in \mathcal{I},j \in \mathcal{J}}\gamma_{k', i, j, k} &   &\notag\\
&   && -\sum_{k\in \mathcal{T}_{t}}\sum_{(k,k')\in E}\sum_{i \in \mathcal{I},j \in \mathcal{J}}\gamma_{k, i, j, k'}\leq w_{t}, && \forall t\in \mathcal{T} \label{job:coflow:dual:a} \\
&  && \alpha_{i, j, k} \geq 0, && \forall k\in \mathcal{K}, \forall i\in \mathcal{I}, \notag\\
&  &&                               &&  \forall j\in \mathcal{J} \label{job:coflow:dual:b} \\
&  && \beta_{i, S}\geq 0,   &&  \forall i\in \mathcal{I}, \forall S\subseteq \mathcal{F}_{i} \label{job:coflow:dual:c} \\
&  && \beta_{j, S}\geq 0,   &&  \forall j\in \mathcal{J}, \forall S\subseteq \mathcal{F}_{j} \label{job:coflow:dual:d} \\
&  && \gamma_{k', i, j, k}\geq 0,   &&  \forall (k', k)\in E,  \notag\\
&  &&                               &&  \forall i\in \mathcal{I}, \forall j\in \mathcal{J}\label{job:coflow:dual:e}
\end{align}
\end{subequations}

Let $\alpha_{i, j, t} = \sum_{k\in \mathcal{T}_{t}} \alpha_{i, j, k}, L_{i,S,t}=\sum_{k\in \mathcal{T}_{t}} L_{i,S,k}$ and $L_{j,S,t}=\sum_{k\in \mathcal{T}_{t}} L_{j,S,k}$ for all $t\in \mathcal{T}$.
Algorithm~\ref{Alg_dual3} in Appendix~\ref{appendix:c} determines the order of job scheduling. Since there are no precedence constraints among the jobs, there is no need to set $\gamma$ to satisfy precedence constraints. We transmit the jobs sequentially, and within each job, the coflows are transmitted in topological-sorting order. As the values of $\gamma$ are all zero, similar to the proof of Theorem~\ref{thm:thm1}, we can obtain the following theorem. Unlike Theorem~\ref{thm:thm1}, this result is not limited to the workload sizes and weights that are topology-dependent in the input instances.

\begin{thm}\label{thm4:thm11}
The proposed algorithm achieves an approximation ratio of $O(\chi)$ for minimizing the total weighted completion time of a given set of multi-stage jobs.
\end{thm}

\section{Experimental Results}\label{sec:Results}
In order to evaluate the effectiveness of the proposed algorithm, this section conducts simulations comparing its performance to that of a previous algorithm. Both synthetic and real traffic traces are used for these simulations, without considering release time. The subsequent sections present and analyze the results obtained from these simulations.

\subsection{Comparison Metrics}
Since the cost of the feasible dual solution provides a lower bound on the optimal value of the coflow scheduling problem, we calculate the approximation ratio by dividing the total weighted completion time achieved by the algorithms by the cost of the feasible dual solution.

\subsection{Randomly Generated Graphs}
In this section, we examine a collection of randomly generated graphs that are created based on a predefined set of fundamental characteristics.
\begin{itemize}
\item \textit{DAG size}, $n$: The number of coflows in the DAG.
\item \textit{Out degree}, $deg$: Out degree of a node.
\item \textit{Parallelism factor}, ($p$)~\cite{Daoud08}: The calculation of the levels in the DAG involves randomly generating a number from a uniform distribution. The mean value of this distribution is $\frac{\sqrt{n}}{p}$. The generated number is then rounded up to the nearest integer, determining the number of levels. Additionally, the width of each level is calculated by randomly generating a number from a uniform distribution. The mean value for this distribution is $p \times \sqrt{n}$, and it is also rounded up to the nearest integer~\cite{Topcuoglu02}. Graphs with a larger value of $p$ tend to have a smaller $\chi$, while those with a smaller value of $p$ have a larger $\chi$.
\item \textit{Workload}, $(W_{min}, W_{max}, L_{min}, L_{max})$~\cite{shafiee2018improved}: 
Each coflow is accompanied by a description $(W_{min}, W_{max}, L_{min}, L_{max})$ that provides information about its characteristics. To determine the number of non-zero flows within a coflow, two values, $w_1$ and $w_2$, are randomly selected from the interval $[W_{min}, W_{max}]$. These values are then assigned to the input and output links of the coflow in a random manner. The size of each flow is randomly chosen from the interval $[L_{min}, L_{max}]$. The construction of all coflows by default follows a predefined distribution based on the coflow descriptions. This distribution consists of four configurations: $(1, 4, 1, 10)$, $(1, 4, 10, 1000)$, $(4, N, 1, 10)$, and $(4, N, 10, 1000)$, with proportions of $41\%$, $29\%$, $9\%$, and $21\%$, respectively. Here, $N$ represents the number of ports in the core.
\end{itemize}

Let $level_k$ denote the level of coflow $k$, and let $Lv(k)=\left\{k'\in \mathcal{K} | level_k < level_{k'}\right\}$ represent the set of coflows that have a higher level than $k$. When constructing a DAG, only a subset of $Lv(k)$ can be selected as successors for each coflow $k$. For coflow $k$, a set of successors is randomly chosen with a probability of $\frac{deg}{\left|Lv(k)\right|}$. To assign weights to each coflow, positive integers are randomly and uniformly selected from the interval $[1, 100]$.

\subsection{Results}
Figure~\ref{fig:ratio4} illustrates the approximation ratio of the proposed algorithm compared to the previous algorithm for synthetic traces. The problem size ranges from 5 to 25 coflows in five network cores, with input and output links set to $N=10$. For each instance, we set $deg=3, p=1$, and $\chi\geq 2$. The proposed algorithms demonstrate significantly smaller approximation ratios than $4\chi+2-\frac{2}{m}$. Furthermore, FDLS outperforms Weaver by approximately $4.7\%$ to $7.5\%$ within this problem size range. Although there are no restrictions on the workload's load and weights being topology-dependent for each instance, we still obtain results lower than $4\chi+2-\frac{2}{m}$. This demonstrates the excellent performance of the algorithm in general scenarios.
\begin{figure}[!ht]
    \centering
        \includegraphics[width=3.8in]{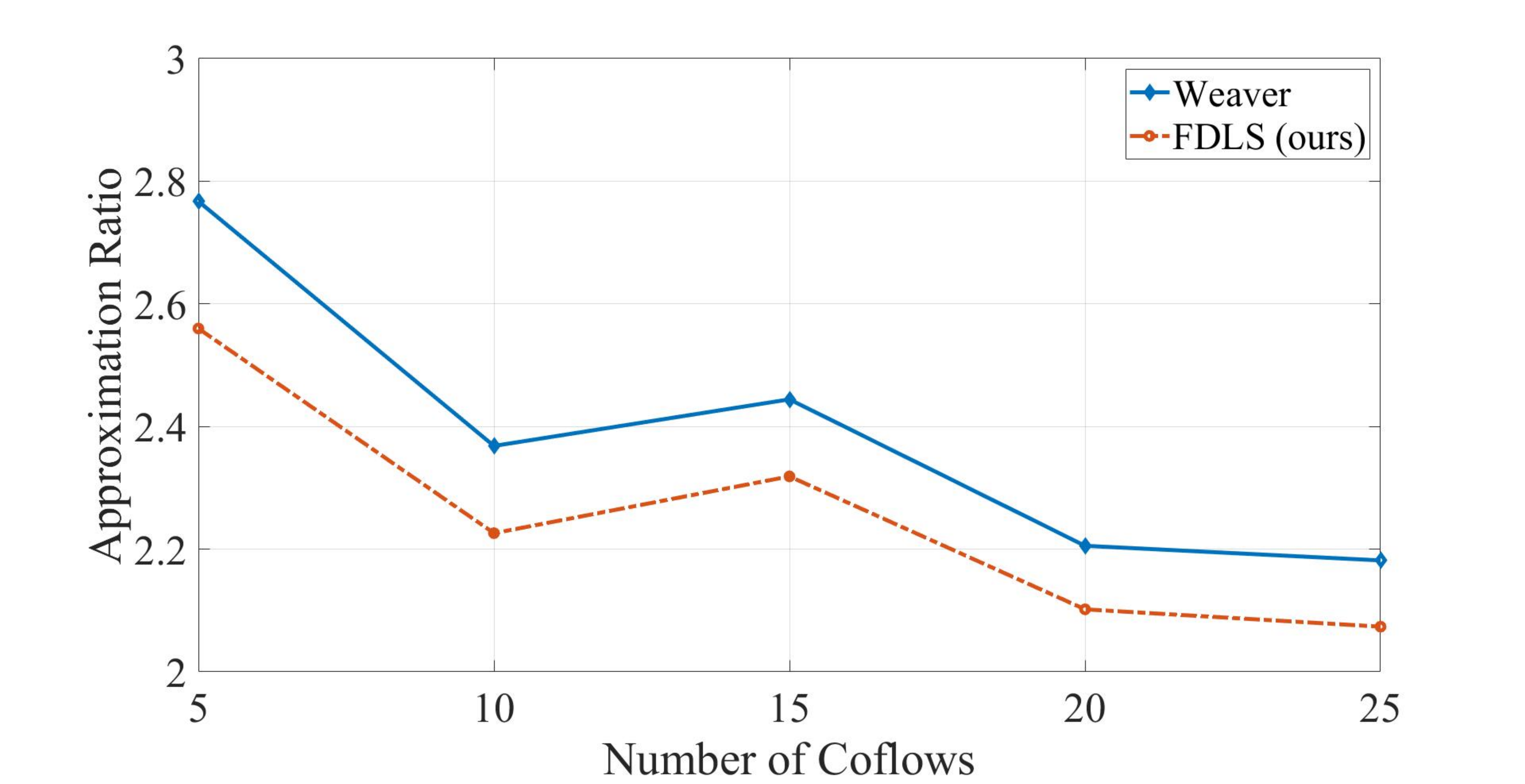}
    \caption{The approximation ratio of synthetic traces between the previous algorithm and the proposed algorithm when
all coflows release at time 0.}
    \label{fig:ratio4}
\end{figure}

The effects of flow density were compared by categorizing the coflows into three instances: dense, sparse, and combined. For each instance, the number of flows was randomly selected from either the range $[N, N^2]$ or $[1, N]$, depending on the specific instance. In the combined instance, each coflow has a $50\%$ probability of being set to sparse and a $50\%$ probability of being set to dense. Figure~\ref{fig:ratio5} illustrates the approximation ratio of synthetic traces for 100 randomly chosen dense and combined instances, comparing the previous algorithm with the proposed algorithm. The problem size consisted of 25 coflows in five network cores, with input and output links set to $N=10$. For each instance, we set $deg=3, p=1$, and $\chi\geq 2$. In the dense case, Weaver achieved an approximation ratio of 2.80, while FDLS achieved an approximation ratio of 2.66, resulting in a $5.12\%$ improvement with Weaver. In the combined case, FDLS outperformed Weaver by $2.52\%$. Importantly, the proposed algorithm demonstrated a greater improvement in the dense case compared to the combined case.

\begin{figure}[!ht]
    \centering
        \includegraphics[width=3.8in]{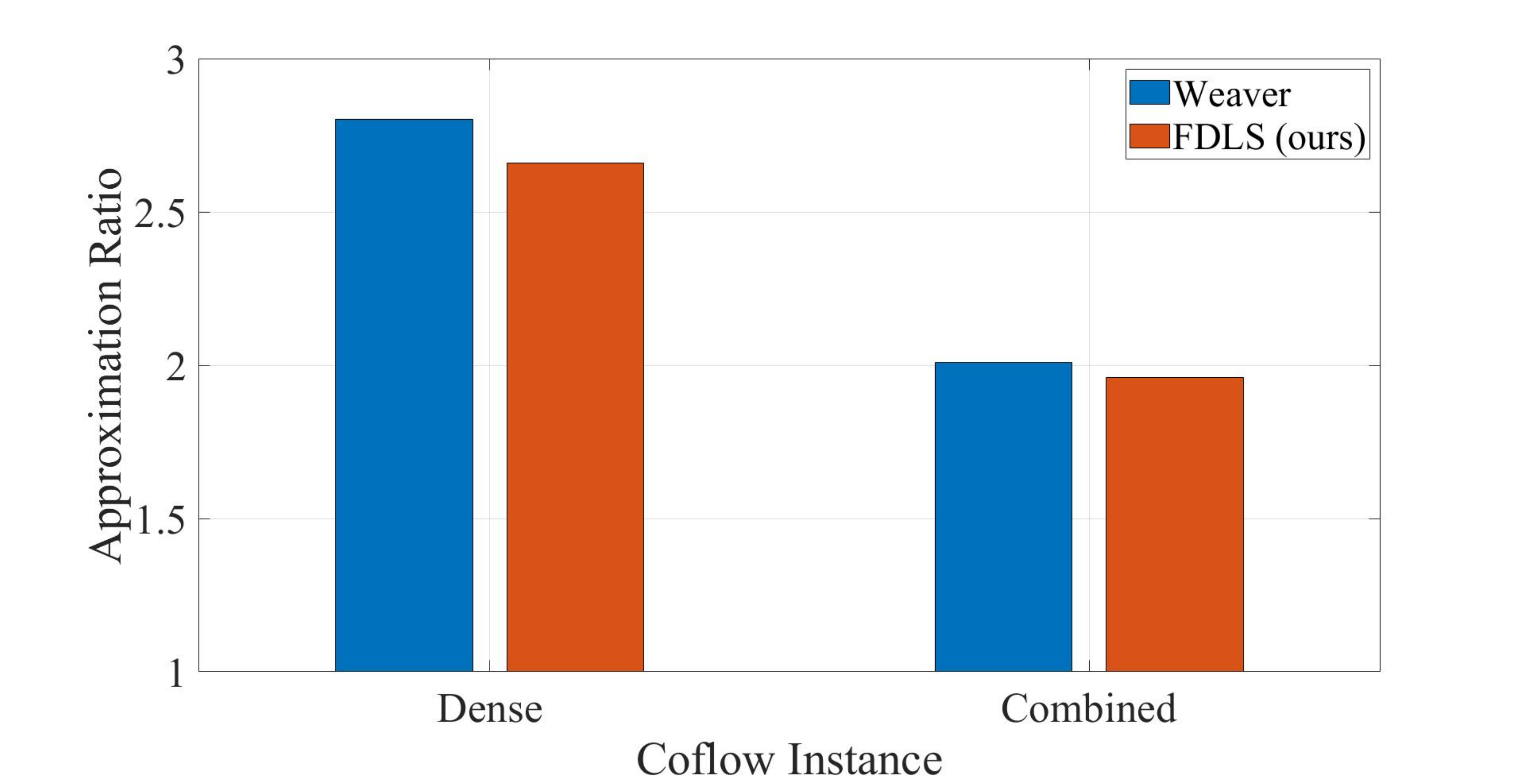}
    \caption{The approximation ratio of synthetic traces between the previous algorithm and the proposed algorithm for different number of coflows when
all coflows release at time 0 for 100 random dense and combined instances.}
    \label{fig:ratio5}
\end{figure}

Figure~\ref{fig:ratio6} illustrates the approximation ratio of synthetic traces for varying numbers of network cores, comparing the previous algorithm to the proposed algorithm when all coflows are released simultaneously at time 0. The problem size consists of 25 coflows distributed across 5 to 25 network cores, with input and output links set to $N=10$. For each instance, we set $deg=3, p=1$, and $\chi\geq 2$. Remarkably, the proposed algorithm consistently achieves significantly smaller approximation ratios compared to the theoretical bound of $4\chi+2-\frac{2}{m}$. As the number of network cores increases, the approximation ratio also tends to increase. This observation can be attributed to the widening gap between the cost of the feasible dual solution and the cost of the optimal integer solution as the number of network cores grows. Consequently, this leads to a notable discrepancy between the experimental approximation ratio and the actual approximation ratio. Importantly, across different numbers of network cores, FDLS outperforms Weaver by approximately $1.79\%$ to $5.30\%$.

\begin{figure}[!ht]
    \centering
        \includegraphics[width=3.8in]{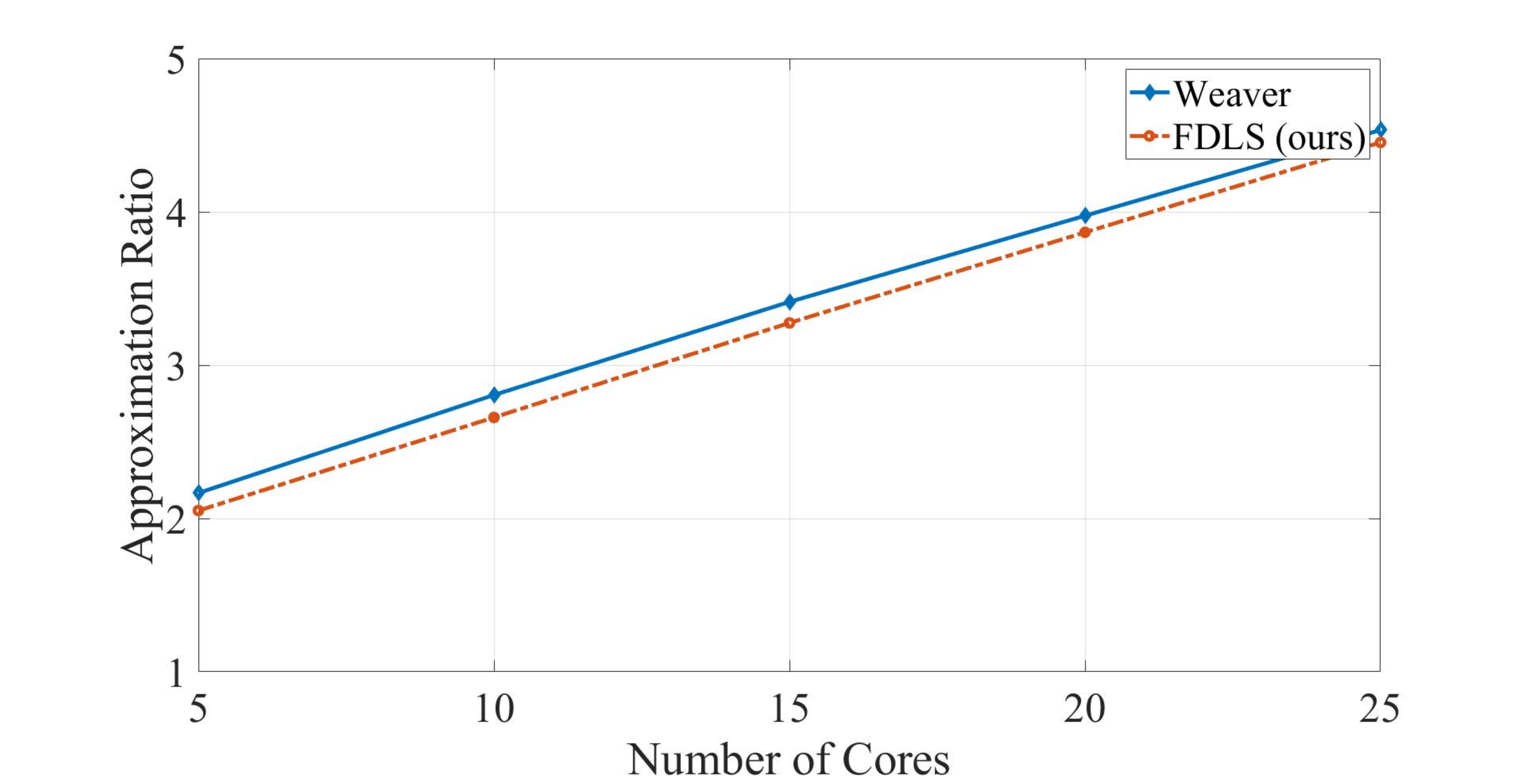}
    \caption{The approximation ratio of synthetic traces between the previous algorithm and the proposed algorithm for different number of network cores when
all coflows release at time 0.}
    \label{fig:ratio6}
\end{figure}

Figure~\ref{fig:ratio6_1} illustrates the approximation ratio of synthetic traces for varying parallelism factor ($p$), comparing the previous algorithm to the proposed algorithm when all coflows are released simultaneously at time 0. The problem size consists of 25 coflows distributed across 5 network cores, with input and output links set to $N=10$. For each instance, we set $deg=3, p=1$, and $\chi\geq 2$. According to our settings, the coflow number of the longest path in the DAG ($\chi$) exhibits an increasing trend as the parallelism factor $p$ decreases. Correspondingly, the approximation ratio also shows an upward trend with a decrease in the parallelism factor $p$. This empirical finding aligns with the theoretical analysis, demonstrating a linear relationship between the approximation ratio and $\chi$.

\begin{figure}[!ht]
    \centering
        \includegraphics[width=3.8in]{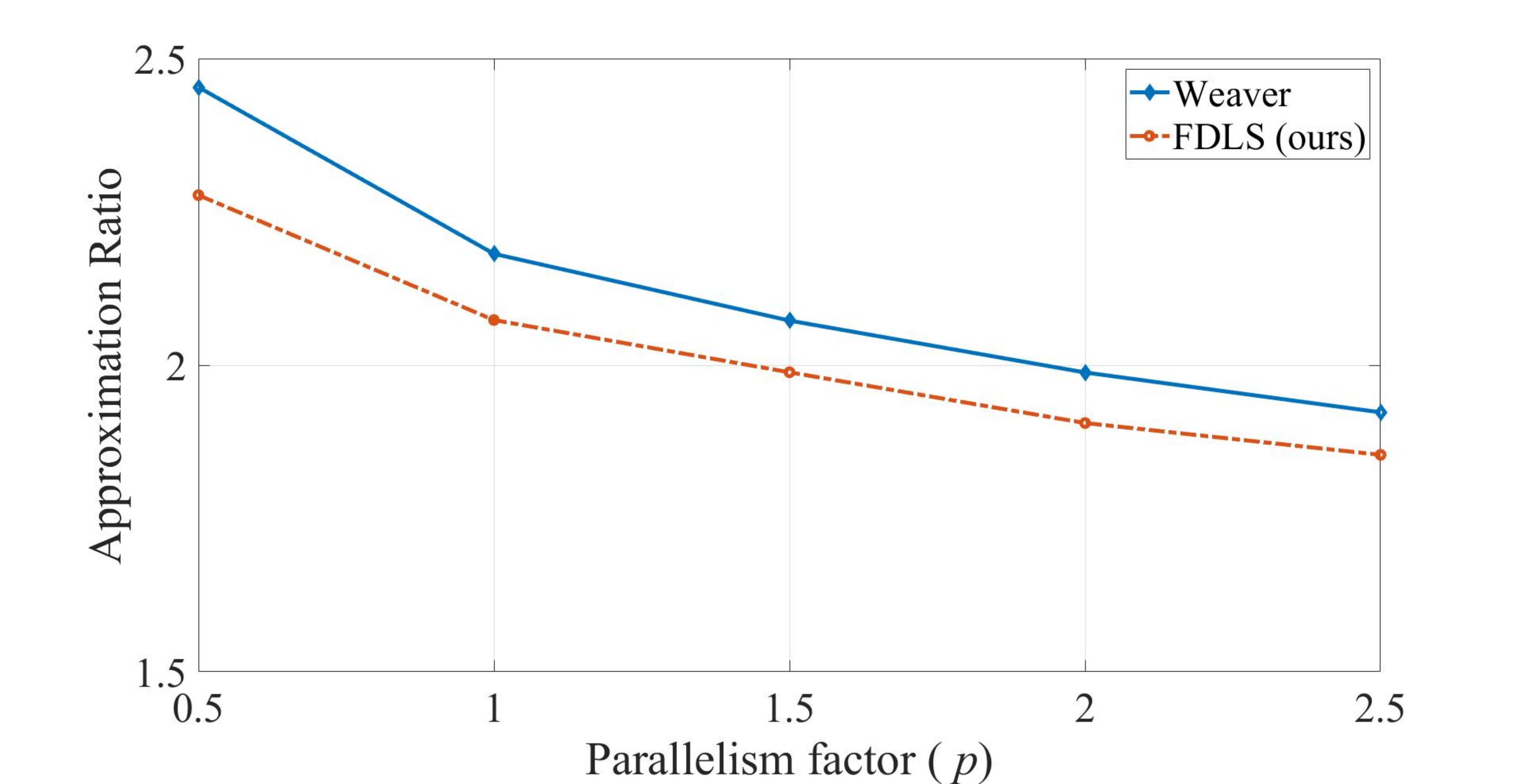}
    \caption{The approximation ratio of synthetic traces between the previous algorithm and the proposed algorithm for different parallelism factor ($p$) when
all coflows release at time 0.}
    \label{fig:ratio6_1}
\end{figure}

We present the simulation results of the real traffic trace obtained from Hive/MapReduce traces captured from Facebook's 3000-machine cluster, consisting of 150 racks. This real traffic trace has been widely used in previous research simulations~\cite{Chowdhury2015, Qiu2015, shafiee2018improved}. The trace dataset comprises a total of 526 coflows. In Figure~\ref{fig:ratio7}, we depict the approximation ratio of the real traces for different thresholds of the number of flows. That is, we apply a filter to the set of coflows based on the condition that the number of flows is equal to or greater than the threshold value. For each instance, we set $deg=3, p=1$, and $\chi\geq 2$. Notably, the proposed FDLS algorithm outperforms the Weaver algorithm by approximately $4.84\%$ to $3.11\%$ across various thresholds. Furthermore, as the number of flows increases, the approximation ratio decreases. This observation is consistent with our previous findings, suggesting a decreasing trend in the approximation ratio as the number of coflows increases.

\begin{figure}[!ht]
    \centering
        \includegraphics[width=3.8in]{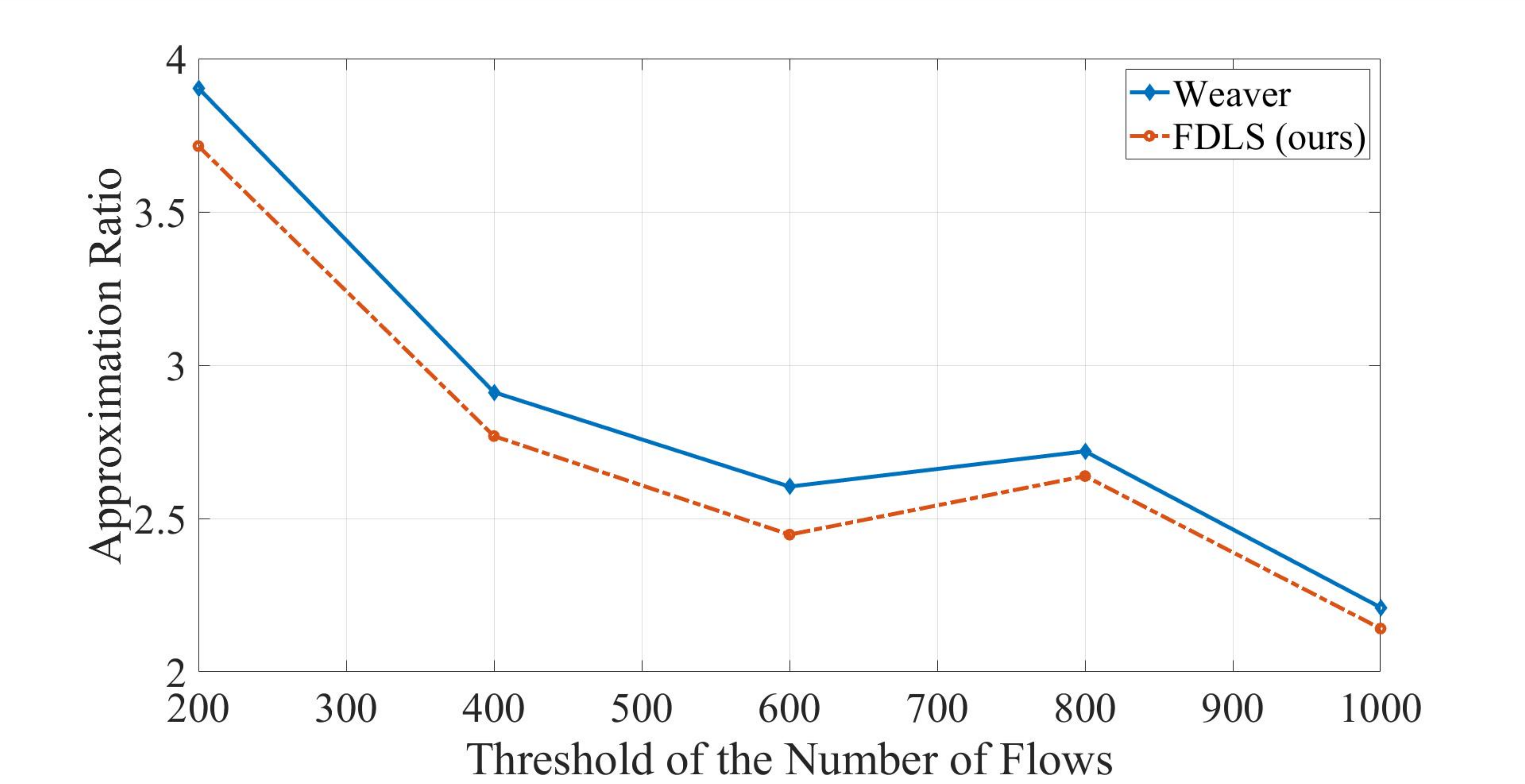}
    \caption{The approximation ratio of real trace between the previous algorithm and the proposed algorithm for different threshold of the number of flows when
all coflows release at time 0.}
    \label{fig:ratio7}
\end{figure}

\section{Concluding Remarks}\label{sec:Conclusion}
This paper focuses on the study the problem of coflow scheduling with release times and precedence constraints in identical parallel networks. The algorithm we propose effectively solves the scheduling order of coflows using the primal-dual method. The primal-dual algorithm has a space complexity of $O(Nn)$ and a time complexity of $O(n^2)$. When considering workload sizes and weights that are topology-dependent in the input instances, our proposed algorithm for the flow-level scheduling problem achieves an approximation ratio of $O(\chi)$. Furthermore, when considering workload sizes that are topology-dependent in the input instances, the algorithm achieves an approximation ratio of $O(R\chi)$. For the coflow-level scheduling problem, the proposed algorithm attains an approximation ratio of $O(m\chi)$ when considering workload sizes and weights that are topology-dependent in the input instances. Moreover, when considering workload sizes that are topology-dependent in the input instances, the algorithm achieves an approximation ratio of $O(Rm\chi)$. In the coflows of multi-stage job scheduling problem, the proposed algorithm achieves an approximation ratio of $O(\chi)$. Although our theoretical results are based on a limited set of input instances, experimental findings show that the results for general input instances outperform the theoretical results, thereby demonstrating the effectiveness and practicality of the proposed algorithm.

\appendix
\section{The Primal-dual Algorithm of Section~\ref{sec:Algorithm1}}\label{appendix:a}
The primal-dual algorithm, presented in Algorithm~\ref{Alg_dual}, draws inspiration from the works of Davis \textit{et al.} \cite{DAVIS2013121} and Ahmadi \textit{et al.} \cite{ahmadi2020scheduling}. This algorithm constructs a feasible schedule iteratively, progressing from right to left, determining the processing order of coflows. Starting from the last coflow and moving towards the first, each iteration makes crucial decisions in terms of increasing dual variables $\alpha$, $\beta$ or $\gamma$. The guidance for these decisions is provided by the dual linear programming (LP) formulation. The algorithm offers a space complexity of $O(Nn)$ and a time complexity of $O(n^2)$, where $N$ represents the number of input/output ports, and $n$ represents the number of coflows.

Consider a specific iteration in the algorithm. At the beginning of this iteration, let $\mathcal{K}$ represent the set of coflows that have not been scheduled yet, and let $k$ denote the coflow with the largest release time. In each iteration, a decision must be made regarding whether to increase dual variables $\alpha$, $\beta$ or $\gamma$. 

If the release time $r_k$ is significantly large, increasing the $\alpha$ dual variable results in substantial gains in the objective function value of the dual problem. On the other hand, if $L_{\mu_1(r)}$ (or $L_{\mu_2(r)}$ if $L_{\mu_2(r)}\geq L_{\mu_1(r)}$) is large, raising the $\beta$ variable leads to substantial improvements in the objective value. Let $\kappa$ be a constant that will be optimized later.

If $r_k>\frac{\kappa\cdot L_{\mu_1(r)}}{m}$ (or $r_k>\frac{\kappa\cdot L_{\mu_2(r)}}{m}$ if $L_{\mu_2(r)}\geq L_{\mu_1(r)}$), the $\alpha$ dual variable is increased until the dual constraint for coflow $k$ becomes tight. Consequently, coflow $k$ is scheduled to be processed as early as possible and before any previously scheduled coflows.
In the case where $r_k\leq \frac{\kappa\cdot L_{\mu_1(r)}}{m}$ (or $r_k\leq\frac{\kappa\cdot L_{\mu_2(r)}}{m}$ if $L_{\mu_2(r)}\geq L_{\mu_1(r)}$), the dual variable $\beta_{\mu_1(r),\mathcal{G}_i}$ (or $\beta_{\mu_2(r),\mathcal{G}_j}$ if $L{\mu_2(r)}\geq L_{\mu_1(r)}$) is increased until the dual constraint for coflow $k'$ becomes tight. 

In this step, we begin by identifying a candidate coflow, denoted as $k'$, with the minimum value of $\beta$. We then examine whether this coflow still has unscheduled successors. If it does, we continue traversing down the chain of successors until we reach a coflow that has no unscheduled successors, which we will refer to as $t_1$.
Once we have identified coflow $t_1$, we set its $\beta$ and $\gamma$ values such that the dual constraint for coflow $t_1$ becomes tight. Moreover, we ensure that the $\beta$ value of coflow $t_1$ matches that of the candidate coflow $k'$.

\begin{algorithm*}
\caption{Permuting Coflows}
    \begin{algorithmic}[1]
		    \STATE $\mathcal{K}$ is the set of unscheduled coflows and initially $K=\left\{1,2,\ldots,n\right\}$
				\STATE $\mathcal{G}_{i}=\left\{(i, j, k)| d_{i,j,k}>0, \forall k\in \mathcal{K}, \forall j\in \mathcal{J} \right\}$
				\STATE $\mathcal{G}_{j}=\left\{(i, j, k)| d_{i,j,k}>0, \forall k\in \mathcal{K}, \forall i\in \mathcal{I} \right\}$
				\STATE $\alpha_{i, j, k}=0$ for all $k\in \mathcal{K}, i\in \mathcal{I}, j\in \mathcal{J}$
				\STATE $\beta_{i, S}= 0$ for all $i\in \mathcal{I}, S\subseteq \mathcal{F}_{i}$
				\STATE $\beta_{j, S}= 0$ for all $j\in \mathcal{J}, S\subseteq \mathcal{F}_{j}$
				\STATE $L_{i,k}=\sum_{j\in \mathcal{J}} d_{i,j,k}$ for all $k\in \mathcal{K}, i\in \mathcal{I}$
				\STATE $L_{j,k}=\sum_{i\in \mathcal{I}} d_{i,j,k}$ for all $k\in \mathcal{K}, j\in \mathcal{J}$
				\STATE $L_{i} = \sum_{k\in \mathcal{K}}L_{i,k}$ for all $i\in \mathcal{I}$
				\STATE $L_{j} = \sum_{k\in \mathcal{K}}L_{j,k}$ for all $j\in \mathcal{J}$		
				\FOR{$r=n, n-1, \ldots, 1$}
				    \STATE $\mu_1(r)=\arg\max_{i\in \mathcal{I}} L_{i}$
				    \STATE $\mu_2(r)=\arg\max_{j\in \mathcal{J}} L_{j}$
						\STATE $k=\arg\max_{\ell\in \mathcal{K}} r_{\ell}$
						\IF{$L_{\mu_1(r)}>L_{\mu_2(r)}$}
                \IF{$r_{k}>\frac{\kappa\cdot L_{\mu_1(r)}}{m}$}
						        \STATE $\alpha_{\mu_1(r), 1, k}=w_{k}-\sum_{i \in \mathcal{I}}\sum_{S\subseteq \mathcal{F}_{i}}\beta_{i,S}L_{i,S,k}-\sum_{j \in \mathcal{J}}\sum_{S\subseteq \mathcal{F}_{j}}\beta_{j,S}L_{j,S,k}-\sum_{k'\in\mathcal{K}|(k',k)\in E}\gamma_{k',k}+\sum_{k'\in\mathcal{K}|(k,k')\in E}\gamma_{k,k'}$
										\STATE $\sigma(r)\leftarrow k$
						    \ELSIF{$r_{k}\leq\frac{\kappa\cdot L_{\mu_1(r)}}{m}$}
						        \STATE $t_1=k'=\arg\min_{k\in \mathcal{K}}\left\{\scriptstyle{\frac{w_{k}-\sum_{i \in \mathcal{I}}\sum_{S\subseteq \mathcal{F}_{i}}\beta_{i,S}L_{i,S,k}-\sum_{j \in \mathcal{J}}\sum_{S\subseteq \mathcal{F}_{j}}\beta_{j,S}L_{j,S,k}-\sum_{k'\in\mathcal{K}|(k',k)\in E}\gamma_{k',k}+\sum_{k'\in\mathcal{K}|(k,k')\in E}\gamma_{k,k'}}{L_{\mu_1(r),\mathcal{G}_{\mu_1(r)},k}}}\right\}$
										\WHILE{$\exists (t_1, t_2)\in E$ and $t_2\in \mathcal{K}$}
											\STATE $t_0=t_1$ and $t_1=t_2$
										\ENDWHILE
										\IF{$t_1\neq k'$}
											\STATE $f_1=\frac{w_{k'}-\sum_{i \in \mathcal{I}}\sum_{S\subseteq \mathcal{F}_{i}}\beta_{i,S}L_{i,S,k'}-\sum_{j \in \mathcal{J}}\sum_{S\subseteq \mathcal{F}_{j}}\beta_{j,S}L_{j,S,k'}-\sum_{k''\in\mathcal{K}|(k'',k')\in E}\gamma_{k'',k'}+\sum_{k''\in\mathcal{K}|(k',k'')\in E}\gamma_{k',k''}}{L_{\mu_1(r),\mathcal{G}_{\mu_1(r)},k'}}$
											\STATE $\gamma_{t_0,t_1}=w_{t_1}-\sum_{i \in \mathcal{I}}\sum_{S\subseteq \mathcal{F}_{i}}\beta_{i,S}L_{i,S,t_1}-\sum_{j \in \mathcal{J}}\sum_{S\subseteq \mathcal{F}_{j}}\beta_{j,S}L_{j,S,t_1}+\sum_{k''\in\mathcal{K}|(t_1,k'')\in E}\gamma_{t_1,k''}-L_{\mu_1(r),\mathcal{G}_{\mu_1(r)},t_1}\cdot f_1$
											\STATE $k'= t_1$
										\ENDIF
										\STATE $\beta_{\mu_1(r),\mathcal{G}_{\mu_1(r)}}=\frac{w_{k'}-\sum_{i \in \mathcal{I}}\sum_{S\subseteq \mathcal{F}_{i}}\beta_{i,S}L_{i,S,k'}-\sum_{j \in \mathcal{J}}\sum_{S\subseteq \mathcal{F}_{j}}\beta_{j,S}L_{j,S,k'}-\sum_{k''\in\mathcal{K}|(k'',k')\in E}\gamma_{k'',k'}+\sum_{k''\in\mathcal{K}|(k',k'')\in E}\gamma_{k',k''}}{L_{\mu_1(r),\mathcal{G}_{\mu_1(r)},k'}}$
										\STATE $\sigma(r)\leftarrow k'$
						    \ENDIF						
						\ELSE
                \IF{$r_{k}>\frac{\kappa\cdot L_{\mu_2(r)}}{m}$}
						        \STATE $\alpha_{1, \mu_2(r), k}=w_{k}-\sum_{i \in \mathcal{I}}\sum_{S\subseteq \mathcal{F}_{i}}\beta_{i,S}L_{i,S,k}-\sum_{j \in \mathcal{J}}\sum_{S\subseteq \mathcal{F}_{j}}\beta_{j,S}L_{j,S,k}-\sum_{k'\in\mathcal{K}|(k',k)\in E}\gamma_{k',k}+\sum_{k'\in\mathcal{K}|(k,k')\in E}\gamma_{k,k'}$
										\STATE $\sigma(r)\leftarrow k$
						    \ELSIF{$r_{k}\leq\frac{\kappa\cdot L_{\mu_2(r)}}{m}$}
						        \STATE $t_1=k'=\arg\min_{k\in \mathcal{K}}\left\{\scriptstyle{\frac{w_{k}-\sum_{i \in \mathcal{I}}\sum_{S\subseteq \mathcal{F}_{i}}\beta_{i,S}L_{i,S,k}-\sum_{j \in \mathcal{J}}\sum_{S\subseteq \mathcal{F}_{j}}\beta_{j,S}L_{j,S,k}-\sum_{k'\in\mathcal{K}|(k',k)\in E}\gamma_{k',k}+\sum_{k'\in\mathcal{K}|(k,k')\in E}\gamma_{k,k'}}{L_{\mu_2(r),\mathcal{G}_{\mu_2(r)},k}}}\right\}$
										\WHILE{$\exists (t_1, t_2)\in E$ and $t_2\in \mathcal{K}$}
											\STATE $t_0=t_1$ and $t_1=t_2$
										\ENDWHILE
										\IF{$t_1\neq k'$}
											\STATE $f_1=\frac{w_{k'}-\sum_{i \in \mathcal{I}}\sum_{S\subseteq \mathcal{F}_{i}}\beta_{i,S}L_{i,S,k'}-\sum_{j \in \mathcal{J}}\sum_{S\subseteq \mathcal{F}_{j}}\beta_{j,S}L_{j,S,k'}-\sum_{k''\in\mathcal{K}|(k'',k')\in E}\gamma_{k'',k'}+\sum_{k''\in\mathcal{K}|(k',k'')\in E}\gamma_{k',k''}}{L_{\mu_2(r),\mathcal{G}_{\mu_2(r)},k'}}$
											\STATE $\gamma_{t_0,t_1}=w_{t_1}-\sum_{i \in \mathcal{I}}\sum_{S\subseteq \mathcal{F}_{i}}\beta_{i,S}L_{i,S,t_1}-\sum_{j \in \mathcal{J}}\sum_{S\subseteq \mathcal{F}_{j}}\beta_{j,S}L_{j,S,t_1}+\sum_{k''\in\mathcal{K}|(t_1,k'')\in E}\gamma_{t_1,k''}-L_{\mu_2(r),\mathcal{G}_{\mu_2(r)},t_1}\cdot f_1$
											\STATE $k'= t_1$
										\ENDIF
										\STATE $\beta_{\mu_2(r),\mathcal{G}_{\mu_2(r)}}=\frac{w_{k'}-\sum_{i \in \mathcal{I}}\sum_{S\subseteq \mathcal{F}_{i}}\beta_{i,S}L_{i,S,k'}-\sum_{j \in \mathcal{J}}\sum_{S\subseteq \mathcal{F}_{j}}\beta_{j,S}L_{j,S,k'}-\sum_{k''\in\mathcal{K}|(k'',k')\in E}\gamma_{k'',k'}+\sum_{k''\in\mathcal{K}|(k',k'')\in E}\gamma_{k',k''}}{L_{\mu_2(r),\mathcal{G}_{\mu_2(r)},k'}}$
										\STATE $\sigma(r)\leftarrow k'$
						    \ENDIF						
						\ENDIF
						\STATE $\mathcal{K}\leftarrow \mathcal{K}\setminus \sigma(r)$
     				\STATE $\mathcal{G}_{i}=\left\{(i, j, k)| d_{i,j,k}>0, \forall k\in \mathcal{K}, \forall j\in \mathcal{J} \right\}$
				    \STATE $\mathcal{G}_{j}=\left\{(i, j, k)| d_{i,j,k}>0, \forall k\in \mathcal{K}, \forall i\in \mathcal{I} \right\}$
    				\STATE $L_{i} = L_{i}-L_{i,\sigma(r)}$ for all $i\in \mathcal{I}$
		    		\STATE $L_{j} = L_{j}-L_{j,\sigma(r)}$ for all $j\in \mathcal{J}$				
				\ENDFOR
   \end{algorithmic}
\label{Alg_dual}
\end{algorithm*}

\section{The Primal-dual Algorithm of Section~\ref{sec:Algorithm2}}\label{appendix:b}
Algorithm~\ref{Alg2_dual} presents the primal-dual algorithm which has a space complexity of $O(Nn)$ and a time complexity of $O(n^2)$, where $N$ represents the number of input/output ports and $n$ represents the number of coflows.

\begin{algorithm*}
\caption{Permuting Coflows}
    \begin{algorithmic}[1]
		    \STATE $\mathcal{K}$ is the set of unscheduled coflows and initially $K=\left\{1,2,\ldots,n\right\}$
				\STATE $\alpha_{i, k}=0$ for all $k\in \mathcal{K}, i\in \mathcal{I}$
				\STATE $\alpha_{j, k}=0$ for all $k\in \mathcal{K}, j\in \mathcal{J}$
				\STATE $\beta_{i, S}= 0$ for all $i\in \mathcal{I}, S\subseteq \mathcal{K}$
				\STATE $\beta_{j, S}= 0$ for all $j\in \mathcal{J}, S\subseteq \mathcal{K}$
				\STATE $L_{i,k}=\sum_{j\in \mathcal{J}} d_{i,j,k}$ for all $k\in \mathcal{K}, i\in \mathcal{I}$
				\STATE $L_{j,k}=\sum_{i\in \mathcal{I}} d_{i,j,k}$ for all $k\in \mathcal{K}, j\in \mathcal{J}$
				\STATE $L_{i} = \sum_{k\in \mathcal{K}}L_{i,k}$ for all $i\in \mathcal{I}$
				\STATE $L_{j} = \sum_{k\in \mathcal{K}}L_{j,k}$ for all $j\in \mathcal{J}$		
				\FOR{$r=n, n-1, \ldots, 1$}
				    \STATE $\mu_1(r)=\arg\max_{i\in \mathcal{I}} L_{i}$
				    \STATE $\mu_2(r)=\arg\max_{j\in \mathcal{J}} L_{j}$
						\STATE $k=\arg\max_{\ell\in \mathcal{K}} r_{\ell}$
						\IF{$L_{\mu_1(r)}>L_{\mu_2(r)}$}
                \IF{$r_{k}>\frac{\kappa\cdot L_{\mu_1(r)}}{m}$}
						        \STATE $\alpha_{\mu_1(r), k}=w_{k}-\sum_{i \in \mathcal{I}}\sum_{S\ni k}\beta_{i,S}L_{i,k}-\sum_{j \in \mathcal{J}}\sum_{S\ni k}\beta_{j,S}L_{j,k}-\sum_{k'\in\mathcal{K}|(k',k)\in E}\gamma_{k',k}+\sum_{k'\in\mathcal{K}|(k,k')\in E}\gamma_{k,k'}$
										\STATE $\sigma(r)\leftarrow k$
						    \ELSIF{$r_{k}\leq\frac{\kappa\cdot L_{\mu_1(r)}}{m}$}
						        \STATE $t_1=k'=\arg\min_{k\in \mathcal{K}}\left\{\frac{w_{k}-\sum_{i \in \mathcal{I}}\sum_{S\ni k}\beta_{i,S}L_{i,k}-\sum_{j \in \mathcal{J}}\sum_{S\ni k}\beta_{j,S}L_{j,k}-\sum_{k'\in\mathcal{K}|(k',k)\in E}\gamma_{k',k}+\sum_{k'\in\mathcal{K}|(k,k')\in E}\gamma_{k,k'}}{L_{\mu_1(r),k}}\right\}$
										\WHILE{$\exists (t_1, t_2)\in E$ and $t_2\in \mathcal{K}$}
											\STATE $t_0=t_1$ and $t_1=t_2$
										\ENDWHILE
										\IF{$t_1\neq k'$}
											\STATE $f_1=\frac{w_{k'}-\sum_{i \in \mathcal{I}}\sum_{S\ni k}\beta_{i,S}L_{i,k}-\sum_{j \in \mathcal{J}}\sum_{S\ni k}\beta_{j,S}L_{j,k}-\sum_{k''\in\mathcal{K}|(k'',k')\in E}\gamma_{k'',k'}+\sum_{k''\in\mathcal{K}|(k',k'')\in E}\gamma_{k',k''}}{L_{\mu_1(r),k'}}$
											\STATE $\gamma_{t_0,t_1}=w_{t_1}-\sum_{i \in \mathcal{I}}\sum_{S\ni t_1}\beta_{i,S}L_{i,t_1}-\sum_{j \in \mathcal{J}}\sum_{S\ni t_1}\beta_{j,S}L_{j,t_1}+\sum_{k''\in\mathcal{K}|(t_1,k'')\in E}\gamma_{t_1,k''}-L_{\mu_1(r),t_1}\cdot f_1$
											\STATE $k'= t_1$
										\ENDIF
										\STATE $\beta_{\mu_1(r),\mathcal{K}}=\frac{w_{k'}-\sum_{i \in \mathcal{I}}\sum_{S\ni k}\beta_{i,S}L_{i,k'}-\sum_{j \in \mathcal{J}}\sum_{S\ni k}\beta_{j,S}L_{j,k'}}{L_{\mu_1(r),k'}}$
										\STATE $\sigma(r)\leftarrow k'$
						    \ENDIF						
						\ELSE
                \IF{$r_{k}>\frac{\kappa\cdot L_{\mu_2(r)}}{m}$}
						        \STATE $\alpha_{\mu_2(r), k}=w_{k}-\sum_{i \in \mathcal{I}}\sum_{S\ni k}\beta_{i,S}L_{i,k}-\sum_{j \in \mathcal{J}}\sum_{S\ni k}\beta_{j,S}L_{j,k}-\sum_{k'\in\mathcal{K}|(k',k)\in E}\gamma_{k',k}+\sum_{k'\in\mathcal{K}|(k,k')\in E}\gamma_{k,k'}$
										\STATE $\sigma(r)\leftarrow k$
						    \ELSIF{$r_{k}\leq\frac{\kappa\cdot L_{\mu_2(r)}}{m}$}
						        \STATE $t_1=k'=\arg\min_{k\in \mathcal{K}}\left\{\frac{w_{k}-\sum_{i \in \mathcal{I}}\sum_{S\ni k}\beta_{i,S}L_{i,k}-\sum_{j \in \mathcal{J}}\sum_{S\ni k}\beta_{j,S}L_{j,k}-\sum_{k'\in\mathcal{K}|(k',k)\in E}\gamma_{k',k}+\sum_{k'\in\mathcal{K}|(k,k')\in E}\gamma_{k,k'}}{L_{\mu_2(r),k}}\right\}$
										\WHILE{$\exists (t_1, t_2)\in E$ and $t_2\in \mathcal{K}$}
											\STATE $t_0=t_1$ and $t_1=t_2$
										\ENDWHILE
										\IF{$t_1\neq k'$}
											\STATE $f_1=\frac{w_{k'}-\sum_{i \in \mathcal{I}}\sum_{S\ni k}\beta_{i,S}L_{i,k}-\sum_{j \in \mathcal{J}}\sum_{S\ni k}\beta_{j,S}L_{j,k}-\sum_{k''\in\mathcal{K}|(k'',k')\in E}\gamma_{k'',k'}+\sum_{k''\in\mathcal{K}|(k',k'')\in E}\gamma_{k',k''}}{L_{\mu_2(r),k'}}$
											\STATE $\gamma_{t_0,t_1}=w_{t_1}-\sum_{i \in \mathcal{I}}\sum_{S\ni t_1}\beta_{i,S}L_{i,t_1}-\sum_{j \in \mathcal{J}}\sum_{S\ni t_1}\beta_{j,S}L_{j,t_1}+\sum_{k''\in\mathcal{K}|(t_1,k'')\in E}\gamma_{t_1,k''}-L_{\mu_2(r),t_1}\cdot f_1$
											\STATE $k'= t_1$
										\ENDIF

										\STATE $\beta_{\mu_2(r),\mathcal{K}}=\frac{w_{k'}-\sum_{i \in \mathcal{I}}\sum_{S\ni k}\beta_{i,S}L_{i,k'}-\sum_{j \in \mathcal{J}}\sum_{S\ni k}\beta_{j,S}L_{j,k'}}{L_{\mu_2(r),k'}}$
										\STATE $\sigma(r)\leftarrow k'$
						    \ENDIF						
						\ENDIF
						\STATE $\mathcal{K}\leftarrow \mathcal{K}\setminus \sigma(r)$
    				\STATE $L_{i} = L_{i}-L_{i,\sigma(r)}$ for all $i\in \mathcal{I}$
		    		\STATE $L_{j} = L_{j}-L_{j,\sigma(r)}$ for all $j\in \mathcal{J}$				
				\ENDFOR
   \end{algorithmic}
\label{Alg2_dual}
\end{algorithm*}

\section{The Primal-dual Algorithm of Section~\ref{sec:Algorithm6}}\label{appendix:c}
Algorithm~\ref{Alg_dual3} determines the order of job scheduling. Since there are no precedence constraints among the jobs, there is no need to set $\gamma$ to satisfy precedence constraints. 

\begin{algorithm*}
\caption{Permuting Jobs}
    \begin{algorithmic}[1]
		    \STATE $\mathcal{T}$ is the set of unscheduled coflows and initially $\mathcal{T}=\left\{1,2,\ldots,n\right\}$
				\STATE $\mathcal{G}_{i}=\left\{(i, j, k)| d_{i,j,k}>0, \forall k\in \mathcal{K}, \forall j\in \mathcal{J} \right\}$ 
				\STATE $\mathcal{G}_{j}=\left\{(i, j, k)| d_{i,j,k}>0, \forall k\in \mathcal{K}, \forall i\in \mathcal{I} \right\}$
				\STATE $\alpha_{i, j, t}=0$ for all $t\in \mathcal{T}, i\in \mathcal{I}, j\in \mathcal{J}$
				\STATE $\beta_{i, S}= 0$ for all $i\in \mathcal{I}, S\subseteq \mathcal{F}_{i}$
				\STATE $\beta_{j, S}= 0$ for all $j\in \mathcal{J}, S\subseteq \mathcal{F}_{j}$
				\STATE $L_{i,t}=\sum_{k\in \mathcal{T}_{t}}\sum_{j\in \mathcal{J}} d_{i,j,k}$ for all $t\in \mathcal{T}, i\in \mathcal{I}$
				\STATE $L_{j,t}=\sum_{k\in \mathcal{T}_{t}}\sum_{i\in \mathcal{I}} d_{i,j,k}$ for all $t\in \mathcal{T}, j\in \mathcal{J}$
				\STATE $L_{i} = \sum_{k\in \mathcal{K}}L_{i,k}$ for all $i\in \mathcal{I}$
				\STATE $L_{j} = \sum_{k\in \mathcal{K}}L_{j,k}$ for all $j\in \mathcal{J}$		
				\FOR{$r=n, n-1, \ldots, 1$}
				    \STATE $\mu_1(r)=\arg\max_{i\in \mathcal{I}} L_{i}$
				    \STATE $\mu_2(r)=\arg\max_{j\in \mathcal{J}} L_{j}$
						\STATE $t=\arg\max_{\ell\in \mathcal{T}} r_{\ell}$
						\IF{$L_{\mu_1(r)}>L_{\mu_2(r)}$}
                \IF{$r_{t}>\frac{\kappa\cdot L_{\mu_1(r)}}{m}$}
						        \STATE $\alpha_{\mu_1(r), 1, t}=w_{t}-\sum_{i \in \mathcal{I}}\sum_{S\subseteq \mathcal{F}_{i}}\beta_{i,S}L_{i,S,t}-\sum_{j \in \mathcal{J}}\sum_{S\subseteq \mathcal{F}_{j}}\beta_{j,S}L_{j,S,t}$
										\STATE $\sigma(r)\leftarrow t$
						    \ELSIF{$r_{t}\leq\frac{\kappa\cdot L_{\mu_1(r)}}{m}$}
						        \STATE $t'=\arg\min_{t\in \mathcal{T}}\left\{\frac{w_{t}-\sum_{i \in \mathcal{I}}\sum_{S\subseteq \mathcal{F}_{i}}\beta_{i,S}L_{i,S,t}-\sum_{j \in \mathcal{J}}\sum_{S\subseteq \mathcal{F}_{j}}\beta_{j,S}L_{j,S,t}}{L_{\mu_1(r),\mathcal{G}_{\mu_1(r)},t}}\right\}$
										\STATE $\beta_{\mu_1(r),\mathcal{G}_{\mu_1(r)}}=\frac{w_{t'}-\sum_{i \in \mathcal{I}}\sum_{S\subseteq \mathcal{F}_{i}}\beta_{i,S}L_{i,S,t'}-\sum_{j \in \mathcal{J}}\sum_{S\subseteq \mathcal{F}_{j}}\beta_{j,S}L_{j,S,t'}}{L_{\mu_1(r),\mathcal{G}_{\mu_1(r)},t'}}$
										\STATE $\sigma(r)\leftarrow t'$
						    \ENDIF						
						\ELSE
                \IF{$r_{t}>\frac{\kappa\cdot L_{\mu_2(r)}}{m}$}
						        \STATE $\alpha_{1, \mu_2(r), t}=w_{t}-\sum_{i \in \mathcal{I}}\sum_{S\subseteq \mathcal{F}_{i}}\beta_{i,S}L_{i,S,t}-\sum_{j \in \mathcal{J}}\sum_{S\subseteq \mathcal{F}_{j}}\beta_{j,S}L_{j,S,t}$
										\STATE $\sigma(r)\leftarrow t$
						    \ELSIF{$r_{t}\leq\frac{\kappa\cdot L_{\mu_2(r)}}{m}$}
						        \STATE $t'=\arg\min_{t\in \mathcal{T}}\left\{\frac{w_{t}-\sum_{i \in \mathcal{I}}\sum_{S\subseteq \mathcal{F}_{i}}\beta_{i,S}L_{i,S,t}-\sum_{j \in \mathcal{J}}\sum_{S\subseteq \mathcal{F}_{j}}\beta_{j,S}L_{j,S,t}}{L_{\mu_2(r),\mathcal{G}_{\mu_2(r)},t}}\right\}$
										\STATE $\beta_{\mu_2(r),\mathcal{G}_{\mu_2(r)}}=\frac{w_{t'}-\sum_{i \in \mathcal{I}}\sum_{S\subseteq \mathcal{F}_{i}}\beta_{i,S}L_{i,S,t'}-\sum_{j \in \mathcal{J}}\sum_{S\subseteq \mathcal{F}_{j}}\beta_{j,S}L_{j,S,t'}}{L_{\mu_2(r),\mathcal{G}_{\mu_2(r)},t'}}$
										\STATE $\sigma(r)\leftarrow t'$
						    \ENDIF						
						\ENDIF
						\STATE $\mathcal{T}\leftarrow \mathcal{T}\setminus \sigma(r)$
     				\STATE $\mathcal{G}_{i}=\left\{(i, j, k)| d_{i,j,k}>0, \forall t\in \mathcal{T}, \forall k\in \mathcal{T}_{t}, \forall j\in \mathcal{J} \right\}$ 
				    \STATE $\mathcal{G}_{j}=\left\{(i, j, k)| d_{i,j,k}>0, \forall t\in \mathcal{T}, \forall k\in \mathcal{T}_{t}, \forall i\in \mathcal{I} \right\}$
    				\STATE $L_{i} = L_{i}-L_{i,\sigma(r)}$ for all $i\in \mathcal{I}$
		    		\STATE $L_{j} = L_{j}-L_{j,\sigma(r)}$ for all $j\in \mathcal{J}$				
				\ENDFOR
   \end{algorithmic}
\label{Alg_dual3}
\end{algorithm*}

\end{document}